\documentclass[sigconf, nonacm]{acmart}

\newcommand\vldbdoi{XX.XX/XXX.XX}
\newcommand\vldbpages{XXX-XXX}
\newcommand\vldbvolume{14}
\newcommand\vldbissue{1}
\newcommand\vldbyear{2020}
\newcommand\vldbauthors{\authors}
\newcommand\vldbtitle{\shorttitle} 
\newcommand\vldbavailabilityurl{URL_TO_YOUR_ARTIFACTS}
\newcommand\vldbpagestyle{plain}

\usepackage[linesnumbered,ruled,noend]{algorithm2e}
\usepackage{subcaption}
\usepackage{enumitem}
\usepackage{graphicx}
\usepackage{amsmath}
\usepackage{amsthm}
\usepackage{array}
\usepackage{tcolorbox}
\usepackage{caption}
\usepackage{multirow}
\usepackage{setspace}
\usepackage{float}

\newtheoremstyle{mydefinition}
  {2pt}   
  {2pt}   
  {\normalfont}  
  {1em}      
  {\scshape} 
  {.}     
  { }     
  {\thmname{#1}\thmnumber{ #2}\thmnote{ (#3)}} 

\theoremstyle{mydefinition}
\newtheorem{definition}{Definition}
\newtheorem{example}{Example}

\newtheorem{lemma}{Lemma}

\renewenvironment{proof}[1][Proof]{\par\textsc{#1.}\hspace{0.3em}\ignorespaces}{\hfill $\square$\par}

\newcommand{\stitle}[1]{\noindent{\bf #1}}

\newcommand{\reffig}[1]{Figure~\ref{fig:#1}}
\newcommand{\refsec}[1]{Section~\ref{sec:#1}}
\newcommand{\refsubsec}[1]{Section~\ref{subsec:#1}}
\newcommand{\reftab}[1]{Table~\ref{tab:#1}}
\newcommand{\refalg}[1]{Algorithm~\ref{alg:#1}}

\newcommand{\refex}[1]{Example~\ref{ex:#1}}

\newcommand{\kw}[1]{{\ensuremath {\mathsf{#1}}}\xspace}

\newcommand{\Roam}[1]{\uppercase\expandafter{\romannumeral #1}}

\newcommand{\algqkhi}{\kw{Query}}
\newcommand{\alggs}{\kw{GreedySearch}}

\newcommand{\ms}{\kw{MSMarco}}
\newcommand{\yt}{\kw{Youtube}}
\newcommand{\dblp}{\kw{DBLP}}
\newcommand{\laion}{\kw{Laion}}
\newcommand{\rg}{iRangeGraph\xspace}

\newcommand{\kh}{KHI\xspace}
\newcommand{\ef}{\textit{ef}\xspace}
\newcommand{\mdim}{\kw{Dim}\xspace}
\newcommand{\mo}{\mathcal{O}}

\begin{document}
\title{Efficient Approximate Nearest Neighbor Search under Multi-Attribute Range Filter}

\author{Yuanhang Yu}
\affiliation{%
  \institution{Tongji University}
  \city{Shanghai}
  \country{China}
}
\email{yuanhangyu@tongji.edu.cn}

\author{Dawei Cheng}
\authornote{Corresponding author.}
\affiliation{%
  \institution{Tongji University}
  \city{Shanghai}
  \country{China}
}
\email{dcheng@tongji.edu.cn}

\author{Ying Zhang}
\affiliation{%
  \institution{Zhejiang Gongshang University}
  \city{Hangzhou}
  \country{China}
}
\email{ying.zhang@zjgsu.edu.cn}

\author{Lu Qin}
\affiliation{%
  \institution{University of Technology Sydney}
  \city{Sydney}
  \country{Australia}
}
\email{lu.qin@uts.edu.au}

\author{Wenjie Zhang}
\orcid{0000-0002-1825-0097}
\affiliation{%
  \institution{The University of New South Wales}
  \city{Sydney}
  \country{Australia}
}
\email{wenjie.zhang@unsw.edu.au}

\author{Xuemin Lin}
\orcid{0000-0002-1825-0097}
\affiliation{%
  \institution{Shanghai Jiao Tong University}
  \city{Shanghai}
  \country{China}
}
\email{xuemin.lin@sjtu.edu.cn}

\begin{abstract}
Nearest neighbor search on high-dimensional vectors is fundamental in modern AI and database systems.  
In many real-world applications, queries involve constraints on multiple numeric attributes, giving rise to range-filtering approximate nearest neighbor search (RFANNS).
While there exist RFANNS indexes for single-attribute range predicates, extending them to the multi-attribute setting is nontrivial and often ineffective.
In this paper, we propose \kh, an index for multi-attribute RFANNS that combines an attribute-space partitioning tree with HNSW graphs attached to tree nodes.
A skew-aware splitting rule bounds the tree height by $O(\log n)$, and queries are answered by routing through the tree and running greedy search on the HNSW graphs.
Experiments on four real-world datasets show that \kh consistently achieves high query throughput while maintaining high recall.  
Compared with the state-of-the-art RFANNS baseline, \kh improves QPS by $2.46\times$ on average and up to $16.22\times$ on the hard dataset, with larger gains for smaller selectivity, larger $k$, and higher predicate cardinality.
\end{abstract}

\maketitle

\pagestyle{\vldbpagestyle}
\begingroup\small\noindent\raggedright\textbf{PVLDB Reference Format:}\\
\vldbauthors. \vldbtitle. PVLDB, \vldbvolume(\vldbissue): \vldbpages, \vldbyear.\\
\href{https://doi.org/\vldbdoi}{doi:\vldbdoi}
\endgroup
\begingroup
\renewcommand\thefootnote{}\footnote{\noindent
This work is licensed under the Creative Commons BY-NC-ND 4.0 International License. Visit \url{https://creativecommons.org/licenses/by-nc-nd/4.0/} to view a copy of this license. For any use beyond those covered by this license, obtain permission by emailing \href{mailto:info@vldb.org}{info@vldb.org}. Copyright is held by the owner/author(s). Publication rights licensed to the VLDB Endowment. \\
\raggedright Proceedings of the VLDB Endowment, Vol. \vldbvolume, No. \vldbissue\ %
ISSN 2150-8097. \\
\href{https://doi.org/\vldbdoi}{doi:\vldbdoi} \\
}\addtocounter{footnote}{-1}\endgroup

\ifdefempty{\vldbavailabilityurl}{}{
\vspace{.3cm}
\begingroup\small\noindent\raggedright\textbf{PVLDB Artifact Availability:}\\
The source code, data, and/or other artifacts have been made available at xxx. 
\endgroup
}

\section{Introduction}
\label{sec:introduction}

Nearest neighbor search (NNS) is a fundamental problem in high-dimensional similarity search that has been extensively studied~\cite{bentley1975multidimensional,indyk1998approximate,li2019approximate,wang2021comprehensive}.
Given a query object $q$ and a set $O$ of data objects in a metric space $(S, \delta)$, 
nearest neighbor search aims to find an object $o^* \in O$ that minimizes $\delta(q,o^*)$.
Due to the curse of dimensionality~\cite{indyk1998approximate}, exact NNS is often impractical on large high-dimensional datasets, motivating extensive work on approximate nearest neighbor (ANN) search~\cite{malkov2014approximate,harwood2016fanng,malkov2018efficient,fu2019fast,jayaram2019diskann,chen2021spann,xie2025graph}.
In modern systems, objects such as images, text, and videos are often represented as high-dimensional vectors.
Consequently, ANN has become a fundamental operation in a wide range of applications, including information retrieval~\cite{liu2007survey,karpukhin2020dense}, machine learning~\cite{lewis2020retrieval,guu2020retrieval}, and vector database systems~\cite{johnson2019billion,wang2021milvus}.

Beyond vanilla ANN, there is increasing demand for ANN queries augmented with constraints on numeric attributes.
This trend is evident both in recent research~\cite{zuo2024serf,xu2024irangegraph} and in industrial systems from Apple~\cite{mohoney2023high}, Zilliz~\cite{wang2021milvus}, and Alibaba~\cite{wei2020analyticdb}, where ANN queries are routinely combined with metadata filters on attributes such as time and popularity.
In these setting, each data object consists of a feature vector together with a tuple of numeric attribute values, and queries are issued as a query vector accompanied by a range predicate over these attributes.
For example, in a scholarly search engine, a user may submit a paper abstract as a query vector and retrieve similar publications whose publication year and citation count both fall within user-specified numeric ranges.
In this paper, we study multi-attribute range-filtering approximate nearest neighbor search (RFANNS) in Euclidean space.

Existing RFANNS indexes~\cite{zuo2024serf,xu2024irangegraph} are primarily tailored to single-attribute range predicates.
The multi-attribute setting is more challenging and is not handled effectively by existing RFANNS indexes.
Specifically, the method in~\cite{zuo2024serf} does not admit a straightforward extension to multiple attributes, while the index proposed in~\cite{xu2024irangegraph} explicitly supports multi-attribute range predicates but suffers from efficiency limitations in this regime.
To address this challenge, one natural approach is to first partition the low-dimensional attribute space and then build graph indexes over the objects associated with each partition.
R-trees are a well-established choice for the partitioning structure, as they are widely used for multidimensional spatial indexing.
However, we argue that the substantial overlap among R-tree nodes degrades the quality of the graph indexes constructed over these partitions.

Motivated by these limitations, we propose \kh, an RFANNS index tailored for multi-attribute numeric range predicates.
Specifically, \kh combines a KD-tree–like partitioning tree $T$ over the attribute space with graph indexes attached to its nodes.
The partitioning tree recursively splits the data along individual attribute dimensions, producing non-overlapping partitions that adapt to attribute distributions.
To keep the tree height under control in the presence of skewed attributes, we adopt a skew-aware splitting rule that only accepts reasonably balanced splits, guided by a balance threshold $\tau$ and a small per-node exclusion set of problematic dimensions.
We show that this rule guarantees an $O(\log n)$ bound on the tree height, where $n$ is the number of data objects.
On top of this tree, each node $p\in T$ stores a single-level HNSW graph $G_p$ built on the object subset $\mo(p)$ covered by its partition.

To answer RFANNS queries efficiently, \kh\ leverages both the partitioning tree and the HNSW graphs.
Given a query $Q = (q,B)$, where $q$ is the query vector and $B$ is a range predicate that constrains the attribute values, we first traverse the tree to identify nodes whose attribute-space partitions intersect $B$ and select a small subset of them as entry points.
Starting from these entry points, the algorithm performs a greedy search.
When a candidate object is expanded, we reconstruct its neighborhood by aggregating incident edges from all HNSW graphs containing that object and retain only in-range neighbors whose attribute tuples satisfy $B$.
Together, these components enable efficient retrieval of the $k$ nearest neighbors satisfying the range predicate $B$.

To construct \kh, we first build the partitioning tree $T$ in the attribute space using the skew-aware splitting strategy.
Given $T$, we then generate single-level HNSW graphs for all tree nodes in a bottom-up fashion: at each leaf node $p$, we build the graph $G_p$ directly over its object set $\mo(p)$, and at internal nodes we reuse and incrementally merge the child graphs instead of rebuilding them from scratch.
To further reduce index construction time, we employ a hybrid parallelization strategy that combines level-wise parallelism across tree nodes with intra-node parallelism during graph merging within individual tree nodes.

In summary, our main contributions are as follows.
\begin{itemize}[topsep=0pt, leftmargin=1em]
\item To the best of our knowledge, this is the first work that specifically studies RFANNS with multi-attribute numeric range predicates in Euclidean space.

\item We propose \kh, a new RFANNS index that combines an attribute-space partitioning tree with HNSW graphs anchored at tree nodes, and design an efficient RFANNS query algorithm that uses only in-range neighbors while still achieving high recall and high query throughput (\refsec{kd_tree}).

\item We conduct an extensive experimental study on four real-world datasets.
The results demonstrate that at recall $0.95$ (or $0.9$), \kh\ improves QPS over a state-of-the-art RFANNS baseline by an average of $2.46\times$, and up to $16.22\times$ on the hard dataset. (\refsec{experiment}).
\end{itemize}

In addition, \refsec{preliminary} presents the problem definition and basic background, \refsec{related_work} provides a detailed discussion of related work, and \refsec{conclusion} concludes the paper. 

\section{Preliminaries}
\label{sec:preliminary}
\subsection{Problem Definition}
\label{subsec:pro_def}
Let $S$ be the $d$-dimensional space and $\delta$: $S\times S\to \mathbb{R}$ be the distance function over $S$.
The $k$-nearest neighbor ($k$-NN) search problem with numerous real-world applications is defined as follows:

\begin{definition}($k$-NN).
\label{def:nns}
Given a finite set $X\subseteq S$, a query vector $q\in S$, a distance function $\delta$ over $S$, and a positive integer $k\le n$, the $k$-nearest neighbor search returns a subset $R\subseteq X$ with $|R|=k$ such that $\forall x\in R$, $\forall y\in X\setminus R$: $\delta(x,q)\le \delta(y,q)$.
\end{definition}

In practice, $k$--NN search is often performed over vectors whose dimensionality ranges from tens to hundreds.
A notable example is retrieval-augmented generation (RAG)~\cite{lewis2020retrieval}.
As dimensionality increases, the curse of dimensionality~\cite{indyk1998approximate} renders exact nearest neighbor search computationally prohibitive.
Therefore, approximate nearest neighbor (ANN) search has attracted significant interest for achieving high efficiency at the cost of only a slight loss in accuracy.
This loss is commonly quantified by the recall metric: given a set $\hat{R}$ of size $k$ returned by an ANN algorithm, recall is defined as $|R \cap \hat{R}|/k$, where $R$ denotes the true $k$ nearest neighbors.

In real-world applications, objects represented as vectors are often associated with structured attributes.
For example, each product on e-commerce platforms can be represented by an image embedding and structured attributes (e.g., price and rating).
Unless otherwise stated, we use vector and embedding interchangeably for the high-dimensional representation of an object.
In this setting, one may desire to retrieve products that are similar in the embedding space while satisfying attribute filters, such as a given price range and a rating above $4$.
This type of search is known as the range-filtering nearest neighbor search (RFNNS).
To formally define RFNNS, we first present a schema-style description of the object set, analogous to a relation schema in relational databases.

\begin{definition}(Object Schema).
\label{def:ob}
An object schema is defined as a pair $(S, A)$, where (1) $S \subseteq \mathbb{R}^d$ denotes the embedding space; (2) $A=(a_1, \dots, a_m)$ specifies the attribute schema, where each $a_i$ is a numeric attribute.
\end{definition}

Given an object schema $(S, A)$, an object $o$ under this schema is represented as a pair $(x, t)$,
where $x \in S$ is a $d$-dimensional vector and $t = (t_1, \dots, t_m)$ is a tuple attribute values conforming to $A$.
An object set $O$ is an instance of $(S, A)$ if every $o \in O$ follows this representation.
Now, we formally define RFNNS as follows:

\begin{definition}(RFNNS).
\label{def:rfann}
Given an object schema $(S, A)$, its instance $O$, a distance function $\delta$ over $S$, and a positive integer $k$, an RFNNS query is defined as $Q=(q,B)$, where $q\in S$ is a query vector and $B$ is a range predicate.
Specifically, $B=\{b_i=[l_i,r_i]\mid i\in\mathcal{J}\}$ where $\emptyset \neq \mathcal{J} \subseteq \{1,\dots,m\}$ specifies a non-empty set of constrained attributes.
An object $o$ satisfies $B$, denoted as $o \models B$, if $l_i \le o.t.t_i \le r_i,\forall i\in\mathcal{J}$.
Let $O_{B}= \{o \in O \mid o \models B\}$ denote the filtered object set of $O$ under $B$.
The RFNNS problem is to find a set $R_Q \subseteq O_B$ with $|R_Q|=k$ such that for any
$o_x=(x,t_x)\in R_Q$ and $o_y=(y,t_y)\in O_B\setminus R_Q$, $\delta(x,q)\le \delta(y,q)$.
\end{definition}

\begin{figure*}[t]
    \centering
    \begin{subfigure}[b]{0.3\textwidth}
        \centering
        \includegraphics[width=0.99\linewidth]{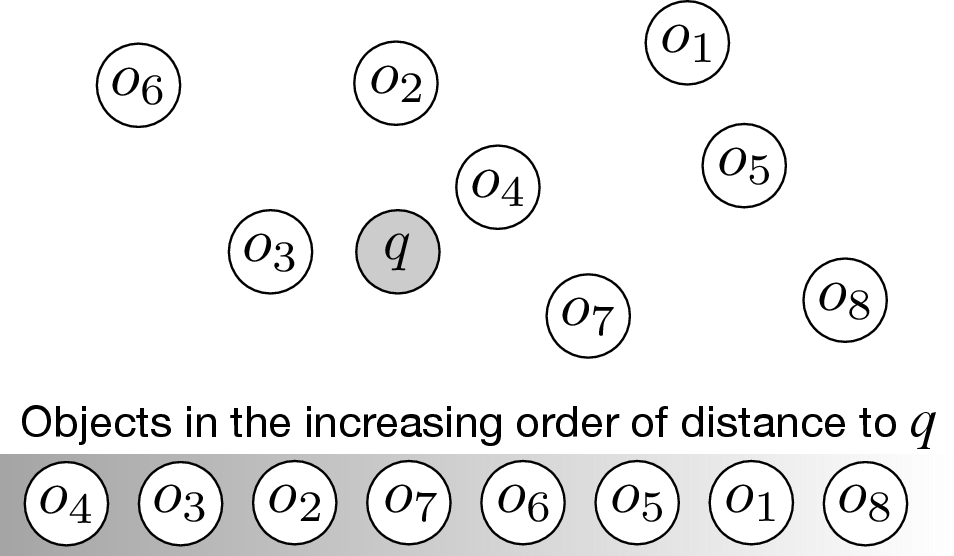}
        \caption{Vectors of an example object set and a query vector $q$.}
        \Description*{}
        \label{fig:emb_obj}
    \end{subfigure}
    \hspace{0.15\textwidth}
    \begin{subfigure}[b]{0.32\textwidth}
        \centering
        \includegraphics[width=0.99\linewidth]{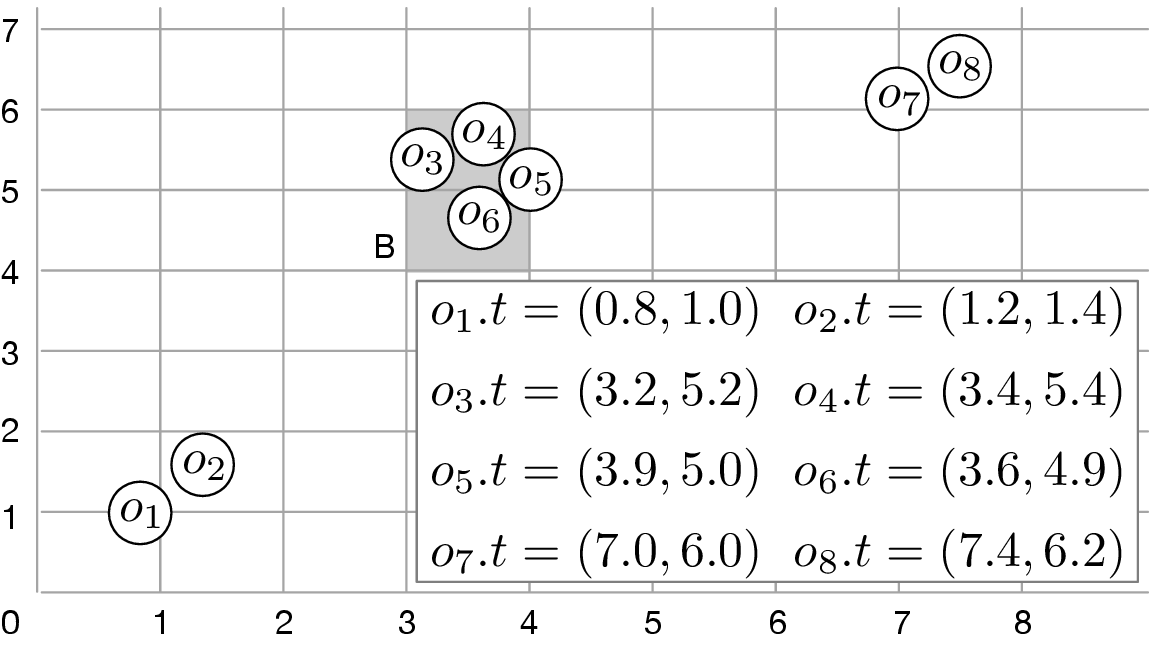}
        \caption{Tuples of an example object set and a range predicate $B$.}
        \Description*{}
        \label{fig:att_obj}
    \end{subfigure}
    \vspace{-1em}
    \caption{An example object set $O$ and an example query $Q$.}
    \label{fig:example_obj}
\end{figure*}

We refer to $|B| = |\mathcal{J}|$ as the cardinality of $B$, i.e., the number of constrained attributes.
To illustrate the above definitions, we consider~\refex{obj}.
\begin{example}
\label{ex:obj}
Given an object schema $(S,A)$ with $m=2$, its instance $O=\{o_i\mid i\in[1,8]\}$ is presented in~\reffig{example_obj}.
Consider an RFNNS query $Q=(q,B)$ with $k=2$, where $q$ and $B=\{b_1=[3.0,4.0],b_2=[4.0,6.0]\}$ are illustrated in~\reffig{emb_obj} and~\reffig{att_obj}, respectively.
$O_{B}=\{o_3,o_4,o_5,o_6\}$ based on their attribute tuples in~\reffig{att_obj}.
For example, $o_3\models B$ because $3.0\le o_3.t.t_1\le 4.0$ and $4.0\le o_3.t.t_2\le 6.0$.
As depicted in~\reffig{emb_obj}, since the vectors of objects $o_3$ and $o_4$ are closer to $q$ than that of $o_5$ and $o_6$, $R_Q=\{o_3, o_4\}$ is returned for the query $Q$.
\end{example}

Since RFNNS faces challenges analogous to exact $k$-NN, i.e., high computational cost at scale, we focus on the approximate variant of RFNNS, referred to as RFANNS, and throughout this paper we use the Euclidean distance as the distance function $\delta$.

\subsection{Graph-based ANN Solutions}
\label{subsec:graph}
We briefly review graph-based ANN methods, which are adopted as core modules of RFANNS methods due to their ability to achieve high recall while probing only a small fraction of the dataset.

Graph-based methods construct a proximity graph $G$ where each vertex $u$ corresponds to a data vector in the dataset.
Given a query vector $q$, they perform a greedy graph traversal, starting from an entry point and iteratively moving to neighbors that are closer to $q$ until no closer candidate can be found.
In practice, the search quality is controlled by a parameter \ef, which limits the size of closest elements maintained during search.
A larger \ef typically explores more vertices, thus improving recall at the cost of higher query time.
Note that \ef is independent of $k$: $k$ defines the query result size, while \ef controls the search breadth.
Existing graph-based methods~\cite{malkov2014approximate,harwood2016fanng,malkov2018efficient,li2019approximate,fu2019fast} primarily differ in their graph construction strategies, which in turn affect the efficiency and effectiveness of greedy navigation.
Among them, NSW~\cite{malkov2014approximate} and HNSW~\cite{malkov2018efficient} have been widely adopted, as they offer an excellent trade-off between search accuracy and efficiency at scale and often exhibit near-logarithmic search path length in practice.

\stitle{Navigable small world graphs.}
NSW builds the proximity graph $G$ by inserting vertices corresponding to data vectors in sequence.
For each inserted vertex $u$, it performs a greedy search on the current graph to obtain an approximate neighbor set $\hat{R}$ of size $M$, and then adds symmetric edges between $u$ and each $v\in \hat{R}$.
Here, the parameter $M$ specifies the number of neighbors linked for each inserted vertex, thus controlling the graph density.
Similarly, the result quality of construction-time search is controlled by the parameter $\ef_b$, which bounds the size of the candidate set explored during insertion.
After inserting all vertices, the construction terminates.
This incremental strategy yields a mix of two critical link types: short-range links that approximate local Delaunay neighborhoods~\cite{aurenhammer1991voronoi} to facilitate fine-grained local exploration around the query vector, and long-range links formed during early insertions to enable fast navigation from distant regions toward the vicinity of the query vector, which together promote small-world navigability and lead to strong empirical efficiency.

\stitle{Hierarchical navigable small world graphs.}
Building upon incremental construction of NSW, HNSW further introduces a hierarchical structure to improve scalability.
Inspired by skip lists, HNSW organizes the proximity graph into multiple layers, where each layer is an NSW-like graph.
The expected number of vertices increases exponentially from the top layer to the bottom layer, enabling a coarse-to-fine greedy traversal.
To reduce redundant connections, HNSW adopts a heuristic neighbor selection strategy inspired by the Relative Neighborhood Graph (RNG)~\cite{toussaint1980relative}.
Specifically, an edge $(u,v)$ is retained only if there exists no neighbor $v'$ of $u$ such that
$\delta(u,v') < \delta(u,v)$ and $\delta(v,v') < \delta(u,v)$; this prunes edges that are shielded by closer intermediate neighbors.
This strategy keeps the graph sparse while ensuring diverse directional connectivity, which is crucial for efficient greedy routing.
To control graph density, HNSW bounds the maximum degree of each vertex, using $M$ for the upper layers and $2M$ for the bottom layer.

\subsection{Existing RFANNS Solutions}
\label{subsec:rfann}
Existing RFANNS solutions~\cite{zuo2024serf,xu2024irangegraph} primarily focus on the single-attribute setting, i.e., $m = 1$, where the range predicate $B$ reduces to a query range $[l, r]$.
In this setting, these methods follow an intuitive paradigm that supports efficient search over the filtered object set $O_B$ under an arbitrary $B$.
Specifically, motivated by the strong empirical performance of HNSW, they seek to enable greedy search over $O_B$ as if an HNSW graph were built on the vectors in $O_B$.
We denote this HNSW graph as $G_{h}(O_B)$ and refer to any HNSW graph built on a subset of $O$ as a filtered HNSW graph.
Since constructing a filtered HNSW graph for every possible query range is prohibitive in both space and build time, these methods typically (i) locate an entry point within $O_B$, and (ii) recover, exactly or approximately, the neighbor list of each visited vertex $u$ on the fly, matching the neighbors that $u$ would have in $G_{h}(O_B)$.
Similarly, each vertex in the graph corresponds to an object in $O$.
Let $O$ be the object set with $|O| = n$.
Without loss of generality, we assume that its one-dimensional attribute domain is discretized and normalized to the integer range $[0, n{-}1]$.

\stitle{SeRF~\cite{zuo2024serf}.}
The discretized attribute domain yields $O(n^2)$ possible query ranges.
SeRF aims to losslessly compress the filtered HNSW graphs associated with these $O(n^{2})$ query ranges, thereby enabling exact on-the-fly reconstruction of the neighbor list for any visited vertex.
This is achieved by compressing neighbor lists across adjacent ranges with the same left boundary.

For a fixed left boundary $l$, as the right boundary $r$ increases from $l$ to $n{-}1$, the membership of a vertex $v$ in the neighbor list of another vertex $u$ is non-reentrant, i.e., $v$ enters and eventually exits without reappearing.
This behavior enables compact encoding via an inner interval $[t_s, t_e]$, such that $v$ is a neighbor of $u$ for all ranges $[l, r]$ with $r \in [t_s, t_e]$.
Moreover, an outer interval $[t_l, t_r]$ specifies the range of left boundaries $l$ for which the inner interval $[t_s, t_e]$ remains valid.
That is, for any $l \in [t_l, t_r]$, $v$ is a neighbor of $u$ in all query ranges $[l, r]$ where $r \in [t_s, t_e]$.

During query processing, SeRF leverages these intervals to exactly recover, on the fly, the neighbor list of each visited vertex for an arbitrary query range $B=[l,r]$, while an entry point in $O_B$ can be efficiently located by ordering objects according to their attribute values.

\stitle{\rg~\cite{xu2024irangegraph}.}
\rg supports approximate, rather than exact, on-the-fly neighbor reconstruction.
To this end, \rg constructs HNSW graphs over a moderate number of ranges, achieving provable bounds on index space and construction time.
Let us refer to a neighbor $v$ of a vertex $u$ as an in-range neighbor with respect to $B$ if $v \models B$; otherwise, $v$ is an out-of-range neighbor.

\rg organizes the attribute domain into subranges using a segment tree with $O(n)$ nodes, where each node corresponds to a subrange of $[0, n{-}1]$.
For each subrange, the corresponding filtered HNSW graph is built.
Given a query range $[l,r]$, \rg reconstructs the neighbor list of each visited vertex $u$ by aggregating its neighbors from the filtered HNSW graphs associated with the segment-tree nodes containing $u$, and retaining only in-range neighbors.
An entry point in $O_B$ can be efficiently obtained from the segment-tree nodes covered by $[l,r]$.

To extend its single-attribute index to multiple attributes, \rg introduces a probabilistic neighbor-reconstruction rule that retains out-of-range neighbors from each filtered HNSW graph with a decaying probability.
This relaxation enlarges the search space and thus improves result quality under multi-attribute range predicates, although the index is built on a single attribute.


\section{Motivation}
\label{sec:r_tree}

\begin{figure*}[t]
    \centering
    \begin{subfigure}[b]{0.4\textwidth}
        \centering
        \includegraphics[width=\linewidth]{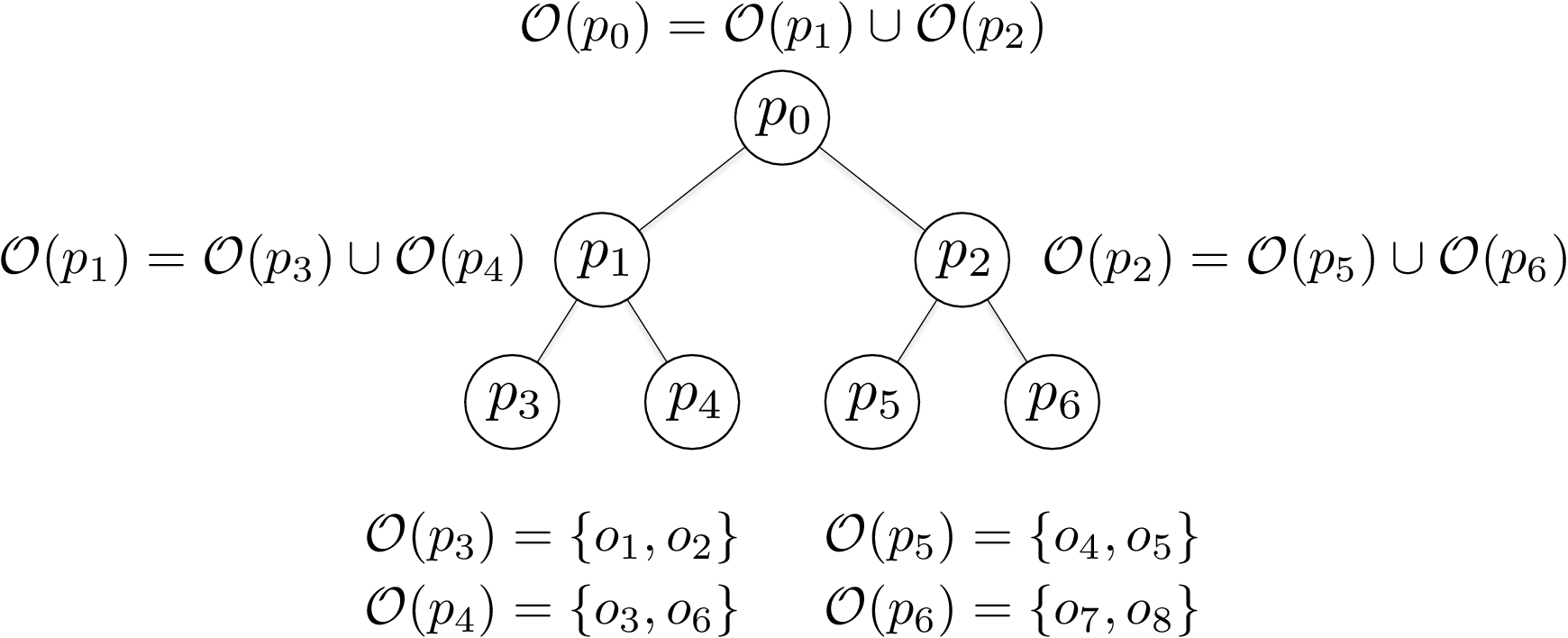}
        \caption{Example R-tree on the example object set.}
        \Description*{}
        \label{fig:RTree}
    \end{subfigure}
    \hspace{0.01\textwidth}
    \begin{subfigure}[b]{0.5\textwidth}
        \centering
        \includegraphics[width=\linewidth]{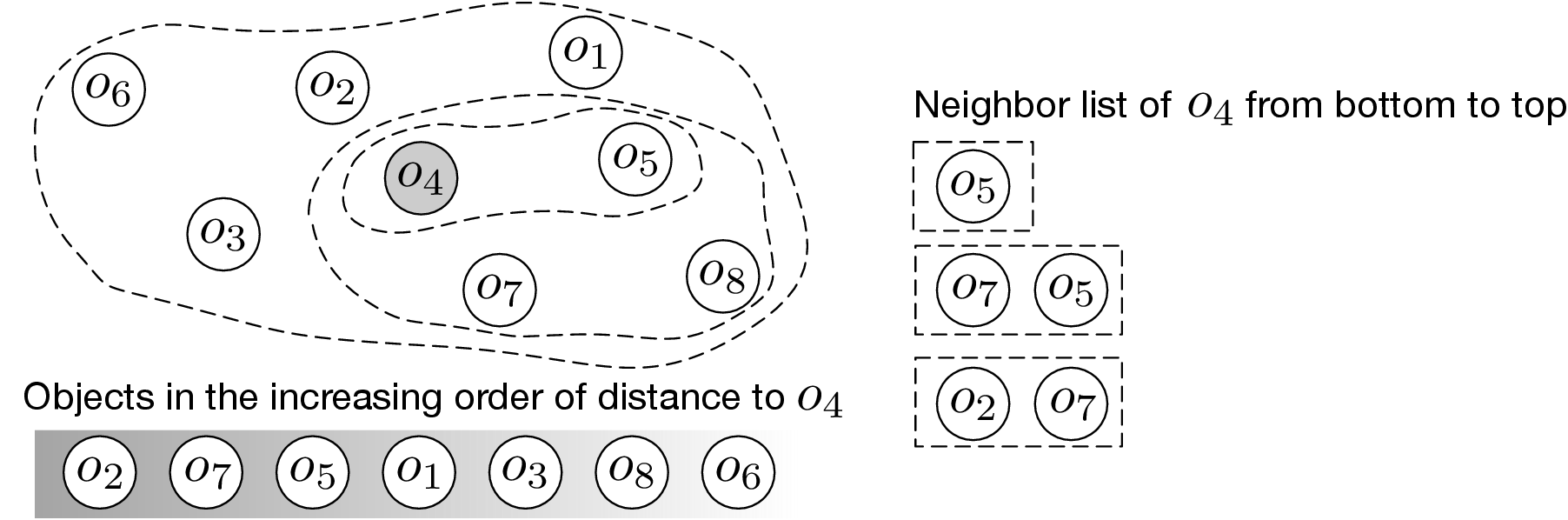}
        \caption{Neighbor lists of $o_4$ across R-tree levels.}
        \Description*{}
        \label{fig:nbr}
    \end{subfigure}
    \vspace{-1em}
    \caption{Running example illustrating neighbor lists of $o_4$ induced by R-tree partitions over the object set $O$.}
    \label{fig:example_RTree}
\end{figure*}

\subsection{Limitations of Existing RFANNS Solutions}
\label{subsec:motivation}
Existing RFANNS methods have demonstrated strong performance in the single-attribute setting.
However, their extensions to multiple attributes are either impractical or ineffective.
We detail these limitations below.

\stitle{Limitations of SeRF.}
SeRF achieves strong performance in the single-attribute setting by enabling exact on-the-fly neighbor reconstruction for each visited vertex.
However, this idea does not scale straightforwardly to multiple attributes.
As the number of attributes in $A$ increases, the space of possible range predicates $B$ can grow exponentially, which would lead to prohibitive index construction time and space overhead.
Therefore, SeRF becomes impractical in the multi-attribute setting.

\stitle{Limitations of \rg.}
\rg supports the multi-attribute setting through a probabilistic rule that allows the search to access out-of-range neighbors with a decaying probability.
This strategy lies between in-filtering and post-filtering: it relaxes the strict in-range constraint of in-filtering, while remaining more selective than post-filtering in exploring all out-of-range neighbors.

However, this relaxation also causes \rg to inherit a key drawback of post-filtering.
For practical multi-attribute RFANNS queries where a subset of attributes are constrained, the filtered set $O_{B}$ is often sparse in $O$, i.e., the selectivity $\sigma = |O_{B}|/|O|$ is low, so the neighbor lists reconstructed on the fly by \rg contain only a small fraction of in-range neighbors.
As a result, the search traverses many out-of-range objects, which slows the refinement of the best-so-far candidates and can even cause the number of visited objects to exceed $|O_{B}|$ when $\sigma$ is sufficiently low.
Our experiments in~\refsubsec{query} confirm that this makes \rg less effective for multi-attribute RFANNS queries, despite its practical utility as an extension of a single-attribute index.

\stitle{Implications.}
The above discussion shows that the multi-attribute setting poses two intertwined challenges that existing RFANNS methods do not sufficiently address.
First, the space of possible range predicates grows rapidly with the number of attributes, making exact on-the-fly neighbor reconstruction for arbitrary $B$ unsustainable.
This suggests approximating the neighbor lists in $G_{h}(O_B)$ using only a moderate number of selectively constructed filtered HNSW graphs.
Second, when the selectivity $\sigma$ is relatively low, out-of-range expansions tend to dominate the search cost.
Thus, an effective index should preserve strong search quality even when traversal is restricted to in-range neighbors.

These observations motivate the following design.
We treat the tuples of objects in $O$ as points in the attribute space and organize them with a spatial partitioning structure.
This structure induces a collection of subsets of $O$, and on each subset we build a corresponding filtered HNSW graph.
By maintaining a moderate number of such graphs, we aim to obtain a close approximate to, for an arbitrary multi-attribute range predicate $B$, the HNSW graph $G_{h}(O_B)$ and thereby achieve a practical balance between query performance and index cost.

\subsection{Limitations of R-tree}
\label{subsec:idx}
Given the above implications, our next step is to identify a spatial index that can partition the attribute space into subsets.
The R-tree is a natural candidate, as it is a classic index for organizing multi-dimensional points and supporting range queries.
We thus examine whether an R-tree over the tuples can induce partitions that are well aligned with our partition-driven graph construction.
As we show next, however, R-tree partitions do not consistently satisfy this requirement.
To avoid confusion, we use node for tree structures and vertex for graph structures.

\stitle{Overlapping MBRs.}
In an R-tree, each internal node stores the minimum bounding rectangles (MBRs) of its children, and the MBRs stored in the same internal node are allowed to overlap in order to maintain balance and compact bounding rectangles.
Thus, the portion of the attribute space indexed by an internal node is represented by overlapping MBRs rather than disjoint cells, and two objects that are close in the attribute space can be assigned to different child nodes of the same parent.

\stitle{Neighbor quality issues.}
In general, distances in the attribute space and in the embedding space are not necessarily aligned: objects that are far apart in the attribute space may still be close in the embedding space, and vice versa.
Since sibling MBRs in an internal R-tree node are allowed to overlap, consider a case where two objects that are close in the attribute space, say $o_u$ and $o_v$, fall into different child nodes $p_1$ and $p_2$ of the same parent $p_0$, while another object $o_w$ that is farther from $o_u$ in the attribute space falls into the same node $p_1$ as $o_u$.
At the same time, $o_w$ may be much closer to $o_u$ than $o_v$ in the embedding space.

Let $\mo(p)$ be the set of objects whose tuples fall within the MBR of the tree node $p$.
For brevity, we write $G_p = G_h(\mo(p))$.
When we build $G_{p_1}$, the neighbors of $o_u$ are restricted to the objects in $\mo(p_1)$.
Since $o_w$ is among the closest objects to $o_u$ in the embedding space within $\mo(p_1)$, $o_w$ becomes a neighbor of $o_u$.
When $G_{p_0}$ is constructed, $o_w$, which is close to $o_u$ in the embedding space, tends to remain in one of the neighbor slots of $o_u$. 
In contrast, $o_v$, which resides in the different child node $p_2$, is pruned due to the limited neighbor-list capacity of $o_u$ and will no longer be considered as a neighbor of $o_u$ in subsequent graph construction.
As a result, the index systematically misses attribute-close neighbors of $o_u$, such as $o_v$, degrading the quality of the local neighborhood around objects.
For a given range predicate $B$, queries may either fail to retrieve certain relevant objects in $O_B$ or reach them only via longer paths, thereby increasing search cost.

\begin{example}
\label{ex:r_tree}
Given the object set $O$ shown in~\reffig{example_obj}, an example R-tree built over $O$ is depicted in~\reffig{RTree}.
Since sibling MBRs in the R-tree may overlap, the closest pair of objects in the attribute space, $o_3$ and $o_4$, is assigned to two different leaf nodes, $p_4$ and $p_5$.
We focus on $o_4$ and its neighbor lists across the R-tree levels.

As shown in~\reffig{nbr}, at node $p_5$ we have $\mo(p_5)=\{o_4,o_5\}$, and the neighbor list of $o_4$ is $\{o_5\}$.
At its parent $p_2$, additional objects such as $o_7$ and $o_8$ are present; because $o_7$ and $o_5$ are closer to $o_4$ than $o_8$ in the embedding space, the neighbor list of $o_4$ becomes $\{o_7, o_5\}$.
At the root $p_0$, more objects including $o_2$ are considered, and the neighbor list of $o_4$ is updated to $\{o_2, o_7\}$.
Note that $o_3$ and $o_4$ first co-occur in the same object set only at $p_0$, i.e., $o_3, o_4 \in \mo(p_0)$.
However, by that time the two neighbor slots of $o_4$ are already occupied by $o_2$ and $o_7$, which are farther from $o_4$ than $o_3$ in the attribute space.
Consequently, the attribute-close neighbor $o_3$ is excluded from the neighbor list of $o_4$ at all levels.
For a range predicate $B$ such that $o_3, o_4 \in O_B$, a graph-based search that starts from $o_4$ and expands along its neighbors may therefore fail to retrieve $o_3$, even though $o_3$ is the closest attribute neighbor of $o_4$.
\end{example}

\stitle{Implications.}
These observations indicate that overlapping partitions in the attribute space distort the neighborhood structure of filtered HNSW graphs and harm search quality (see~\cite{khi-tech-report} for details).
Thus, for RFANNS, spatial indexes that induce non-overlapping partitions in the attribute space are therefore preferable.

\section{Our Solution}
\label{sec:kd_tree}


\begin{figure*}[t]
    \centering
    \begin{subfigure}[b]{0.35\textwidth}
        \centering
        \includegraphics[width=\linewidth]{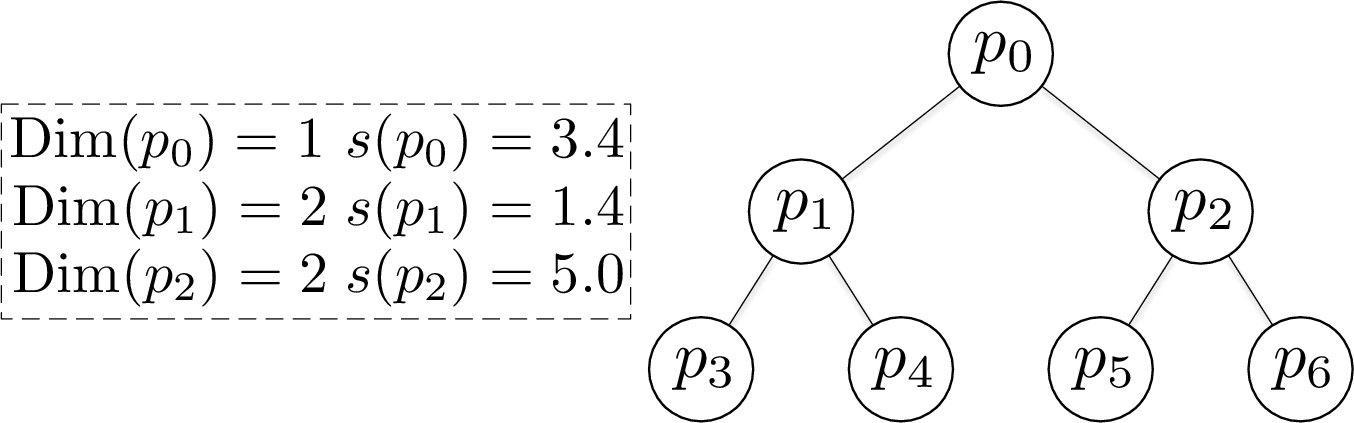}
    \caption{The partitioning tree of \kh on the object set $O$.}
    \Description*{}
    \label{fig:KDTree}
    \end{subfigure}
    \hspace{0.02\textwidth}
    \begin{subfigure}[b]{0.6\textwidth}
        \centering
        \includegraphics[width=\linewidth]{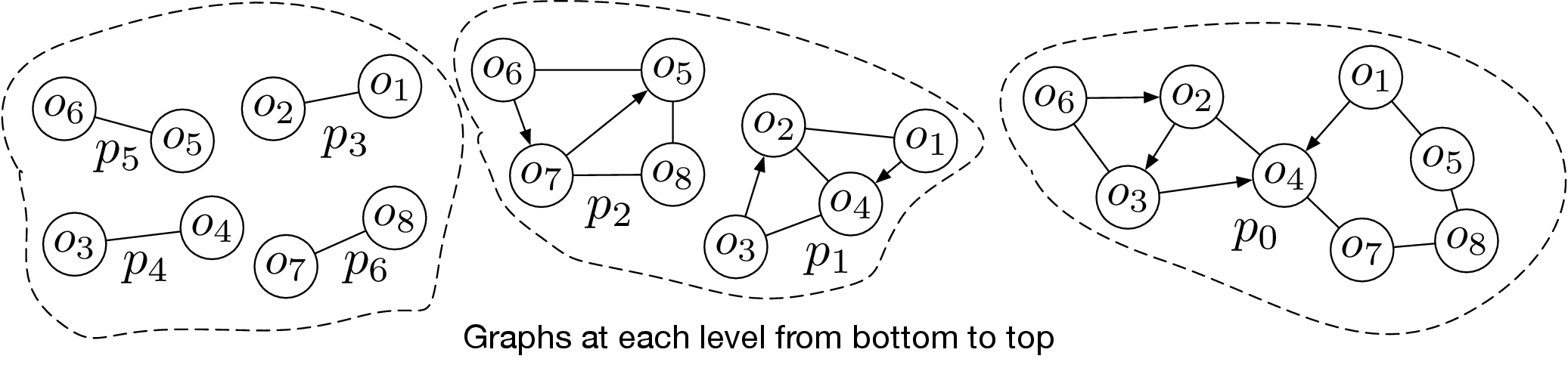}
        \caption{Graphs of \kh on the object set $O$.}
        \Description*{}
        \label{fig:hnsw_graphs}
    \end{subfigure}
    \vspace{-1em}
    \caption{Example of the \kh index over the object set $O$.}
    \label{fig:example_khi}
\end{figure*}

\subsection{Index Overview}
\label{subsec:kdtree}
We propose a KD-tree–HNSW hybrid index, denoted by \kh, which integrates a KD-tree–like partitioning tree on attribute tuples with filtered HNSW graphs over vectors.

\stitle{Attribute-space partitioning.}
\kh maintains a binary partitioning tree $T$ on the attribute tuples of $O$.
Each node $p$ in $T$ is associated with (1) an axis-aligned hyper-rectangle $\mathcal{R}(p)$ in the $m$-dimensional attribute space, and (2) an object subset $\mo(p) \subseteq O$, defined as
$\mo(p) = \{o \in O \mid o.t \in \mathcal{R}(p)\}$.
For the root node $p_r$, $\mathcal{R}(p_r)$ covers all attribute tuples and $\mo(p_r) = O$.
A node $p$ is a leaf if $|\mo(p)| \le c_l$, where $c_l$ is the leaf-capacity parameter (set to $2$ in our implementation).

When splitting an internal node $p$, \kh proceeds as follows.
\kh first chooses a splitting dimension $\mdim(p) \in \{1,\dots,m\}$ and a split value $s(p)$.
Along dimension $\mdim(p)$, we split the region $\mathcal{R}(p)$ at $s(p)$ into two disjoint child regions $\mathcal{R}(p_\ell)$ and $\mathcal{R}(p_r)$.
Objects with $o.t.t_{\mdim(p)} \le s(p)$ are assigned to the left child, and those with $o.t.t_{\mdim(p)} > s(p)$ to the right.
To keep the tree reasonably balanced, we follow a round-robin rule for choosing the splitting dimension:
$\mdim(p)$ is set to the next dimension after the splitting dimension of its parent node (wrapping around after $m$).
Given $\mdim(p)$, we collect the values $\{ o.t.t_{\mdim(p)} \mid o \in \mo(p) \}$ and set $s(p)$ to their median.
However, this median-based choice of $s(p)$ may still lead to a highly unbalanced partition on skewed data.
To mitigate this, we enforce a balance threshold $\tau > 1$.
Let $n_{L}$ and $n_{R}$ denote the numbers of objects assigned to the left and right child of $p$, respectively.
If $\frac{\max\{n_{L}, n_{R}\}}{\min\{n_{L}, n_{R}\}} \ge \tau$, we regard dimension $\mdim(p)$ as skewed at $p$, , exclude this dimension from consideration, and retry the split using the next available dimension.
Once excluded at $p$, a dimension is no longer used as a splitting dimension for any descendant of $p$.

\begin{lemma}
\label{lem:height}
The height of $T$ is $O(\log \frac{n}{c_l})$.
\end{lemma}
\begin{proof}[Proof sketch]
Let $N = |\mo(p)|$ for an internal node $p$, and assume $n_{L} \ge n_{R}$.
From $\frac{n_{L}}{n_{R}} < \tau$ and $n_{L} + n_{R} = N$, we obtain $n_{L} < \frac{\tau}{\tau+1} N$.
Hence, after each recursive split, the larger child has size at most $\rho N$, where $\rho = \frac{\tau}{\tau+1} < 1$.
Starting from $n$ objects at the root, the size of any node is at most $\rho^h n$ after $h$ levels.
The recursion stops once $|\mo(p)| \le c_l$, which implies $\rho^h n \le c_l$ and therefore $h \le \log_{1/\rho}\frac{n}{c_l} = O(\log \frac{n}{c_l})$.
\end{proof}

\stitle{Embedding-space graphs.}
For each node $p$ in $T$, \kh builds a filtered HNSW graph $G_p$.
The vertex set of $G_p$ consists of the objects in $\mo(p)$, and edges are constructed by the standard HNSW procedure.
In our implementation, each $G_p$ is a single-level HNSW graph ($L = 1$): the KD-tree–like hierarchy already provides the coarse-to-fine search structure, so additional HNSW layers are unnecessary.
We use $M$ to bound the maximum degree of each vertex in all single-level HNSW graphs.

\begin{example}
Consider the object set $O$ in~\reffig{example_obj} and the leaf capacity is $c_l = 2$.
The partitioning tree $T$ of \kh on $O$ is illustrated in~\reffig{KDTree}.
At the root $p_0$, \kh splits on the first attribute, i.e., $\mdim(p_0)=1$, with split value $s(p_0)=3.4$, the median of the first-attribute values in $\mo(p_0)$.
Objects with $o.t.t_1\le 3.4$ are assigned to the left child $p_1$, and the rest to the right child $p_2$, resulting in $|\mo(p_1)| = |\mo(p_2)| = 4$.
For $p_1$ and $p_2$, the splitting dimension advances to the second attribute, i.e., $\mdim(p_1)=\mdim(p_2)=2$, and the split values are set to $s(p_1)=1.4$ and $s(p_2)=5.0$, respectively.
Both splits are balanced, so no dimension is excluded, and all children $p_3$–$p_6$ satisfy $|\mo(p)| \le c_l$ and become leaves.

For each node $p$ in $T$, \kh builds a node-level graph $G_p$.
\reffig{hnsw_graphs} shows these graphs for $M = 2$: four attached to the leaves $p_3$–$p_6$ (each over two objects), two for $p_1$ and $p_2$ (each over four objects), and one for the root $p_0$ over all objects.
\end{example}

From the above example, we see that under the partitioning strategy of \kh, each leaf-level graph is built over a small subset of objects that are close in the attribute space.
For instance, $o_3$ and $o_4$, the closest pair in the attribute space, fall into the same leaf node and are linked as mutual neighbors in $G_{p_4}$.
When neighbor reconstruction is performed for a range predicate $B$, the locality of leaf-level graphs in attribute space makes them more likely to contribute in-range neighbors, which can, in many cases, help the reconstructed neighbor lists of visited objects more closely approximate those in the HNSW graph $G_{h}(O_B)$.

\begin{lemma}
The space complexity of \kh is $O(nM\cdot\log \frac{n}{c_l})$.
\end{lemma}
\begin{proof}
By Lemma~\ref{lem:height}, the tree $T$ has height $O(\log \frac{n}{c_l})$.
Since each object belongs to exactly one node at each level of $T$, it participates in $O(\log \frac{n}{c_l})$ node-level filtered HNSW graphs in total.
In each single-level filtered HNSW graph, every object stores at most $M$ neighbors in its adjacency list.
Hence, the total space required by all filtered HNSW graphs is
$O(n M \cdot \log \frac{n}{c_l})$.
The tree $T$ contains $O(n)$ nodes and therefore occupies $O(n)$ space.
Thus, the overall space complexity of \kh is
$O(n M \cdot\log \frac{n}{c_l})$.
\end{proof}

\subsection{Query Processing}
\label{subsec:query}
Given an RFANNS query $Q = (q, B)$, \kh performs a greedy search over $O_B$.
Starting from an entry point in $O_B$, the search iteratively expands from the current object by aggregating its in-range neighbors from the filtered HNSW graphs.

\begin{algorithm}[t]
\caption{\kw{RangeFilter}$(T, B, c_e)$}
\label{alg:range_filter}
\KwIn{Partitioning tree $T$ of \kh with root $p_r$, range predicate $B$, and entry-point budget $c_{e}$}
\KwOut{A set $\mathcal{P}$ of entry points}
$\mathcal{C},\mathcal{P} \gets \emptyset;\mathsf{stack} \gets \{(p_r, \emptyset)\}$\;
\While{$\mathsf{stack} \neq \emptyset \land |\mathcal{C}| < c_e$}{
  $(p, D) \gets \mathsf{stack}.\textsf{pop}()$\;
  $D \gets D \cup \mathsf{BL}(p)$\;
  \lIf{$|D| = m$}{
    $\mathcal{C} \gets \mathcal{C} \cup \{p\};$ \textbf{continue}
  }
  \lIf{$p$ is a leaf}{\textbf{continue}}
  \If{$\mdim(p) \in D$}{
    \lForEach{child $p_c$ of $p$}{
      $\mathsf{stack}.\textsf{push}(p_c, D)$
    }
  }
  \Else{
    \ForEach{child $p_c$ of $p$}{
      $[l_c, r_c] \gets \pi_{\mdim(p)}(\mathcal{R}(p_c))$\;
      \lIf{$l_c > r_{\mdim(p)} \lor r_c < l_{\mdim(p)}$}{
        \textbf{continue}
      }
      \ElseIf{$l_c \ge l_{\mdim(p)} \land r_c \le r_{\mdim(p)}$}{
        $\mathsf{stack}.\textsf{push}(p_c, D \cup \{\mdim(p)\})$\;
      }
      \lElse{
        $\mathsf{stack}.\textsf{push}(p_c, D)$
      }
    }
  }
}
\ForEach{$p\in \mathcal{C}$}{
    \ForEach{$o\in \mo(p)$}{
        \lIf{$o.t\models B$}{
            $\mathcal{P}\gets \mathcal{P} \cup \{o\};$ \textbf{break}
        }
    }
}
\Return $\mathcal{P}$\;
\end{algorithm}

\stitle{Entry point selection.}
With a range predicate $B$ and an entry-point budget $c_e$ as input, \refalg{range_filter} constructs a set $\mathcal{P}$ of at most $c_e$ entry points for the subsequent greedy search.
It performs a depth-first traversal of $T$ to collect a small set of candidate nodes $\mathcal{C}=\{p\}$ such that objects satisfying $B$ can be efficiently located within their object sets $\mo(p)$.
It then scans $\mo(p)$ for each tree node $p \in \mathcal{C}$, adding at most one object $o \models B$ from each $\mo(p)$ to $\mathcal{P}$.
For notational convenience, any attribute not constrained in $B$ is treated as having the trivial query range $(-\infty, +\infty)$.

At initialization, the algorithm creates an empty candidate–node set $\mathcal{C}$, an empty entry–point set $\mathcal{P}$, and a stack seeded with the root node $p_r$ and an empty set of covered dimensions (Line~1).
In Lines 2–15, the algorithm performs a depth-first traversal over $T$ while the stack is non-empty and $|\mathcal{C}| < c_e$.
In each step, it pops a pair $(p, D)$ from the stack (Line 3).
Here, for each node $p$, the set $D$ consists of attribute dimensions $a_i \in A$ that fall into two categories:
(1) dimensions already fully covered by the range predicate, i.e., $\pi_i(\mathcal{R}(p)) \subseteq b_i$, where $\pi_i(\mathcal{R}(p))$ denotes the projection of the rectangle $\mathcal{R}(p)$ onto dimension $a_i$; and
(2) dimensions identified as severely skewed and therefore excluded from further splitting.
The algorithm then augments $D$ with the set of excluded dimensions $\mathsf{BL}(p)$ (Line~4).
Intuitively, dimensions in $\mathsf{BL}(p)$ have been identified as highly skewed at $p$, so further partitioning along them is discouraged.
If all dimensions are covered, i.e., $|D| = m$, the node $p$ is added to $\mathcal{C}$ and its subtree is not explored (Line~5).
If $p$ is a leaf and $|D| < m$, the node is skipped and the traversal continues with the remaining stack entries (Line~6), since not all dimensions in $\mathcal{R}(p)$ w.r.t. $B$ are sufficiently refined.
Otherwise, the algorithm proceeds to the children of $p$.

\begin{algorithm}[t]
\caption{$\kw{ReconsNbr}(\mathcal{I}, o, B, c_n, \mathsf{visited})$}
\label{alg:cons_nbr}
\KwIn{Index \kh $\mathcal{I}$, object $o$, range predicate $B$, neighbor budget $c_n$, and visited set $\mathsf{visited}$}
\KwOut{Reconstructed neighbor list $\mathcal{N}(o)$}
$\mathcal{N}(o)\gets \emptyset$\;
\For{$\ell \gets \mathcal{I}.T.h-1$ \kw{downto} $0$}{
    $p \gets \text{node at level }\ell\text{ such that }o \in \mo(p)$\;
    \ForEach{\kw{neighbor} $o_v$ \kw{of} $o$ \kw{in} $\mathcal{I}.G_p$}{
        \lIf{$\mathsf{visited}[o_v]$}{\textbf{continue}}
        $\mathsf{visited}[o_v] \gets \kw{true}$\;
        \lIf{$\lnot (o_v\models B)$}{\textbf{continue}}
        $\mathcal{N}(o) \gets \mathcal{N}(o) \cup \{o_v\}$\;
        \lIf{$|\mathcal{N}(o)| = c_n$}{
            \Return $\mathcal{N}(o)$
        }
    }
}
\Return $\mathcal{N}(o)$\;
\end{algorithm}

If $\mdim(p)\in D$, then the splitting dimension has already been treated as covered by $B$, and no further refinement along $\mdim(p)$ is needed. Consequently, all children of $p$ are pushed onto the stack with the same set $D$ (Lines 7–8).
If $\mdim(p)\notin D$, then for each child $p_c$ of $p$, the algorithm compares the interval $[l_c,r_c]=\pi_{\mdim(p)}(\mathcal{R}(p_c))$ with the range $[l_{\mdim(p)}, r_{\mdim(p)}]$ of $B$ on dimension $\mdim(p)$:
If $[l_c,r_c]$ does not intersect $[l_{\mdim(n_i)}, r_{\mdim(n_i)}]$, child $p_c$ is discarded (Line~12).
If $[l_c,r_c]$ is fully contained in $[l_{\mdim(n_i)}, r_{\mdim(n_i)}]$, then $p_c$ is pushed with the updated set $D \cup \{\mdim(p)\}$ (Lines~13–14); otherwise, $p_c$ is pushed with the original $D$ (Lines~15).
This depth-first exploration continues until the stack becomes empty or the number of candidate nodes reaches the budget $c_e$.
For each candidate node $p \in \mathcal{C}$, the algorithm designates the first object satisfying $B$ from $\mo(p)$ as an entry point (Lines~16–18).
The resulting set $\mathcal{P} \subseteq O_B$ is well spread over $O_B$ in the attribute space, providing diverse starting points for the subsequent greedy search and mitigating the risk of being trapped in poor local optima.

\stitle{On-the-fly neighbor reconstruction.}
\refalg{cons_nbr} reconstructs neighbors of an object $o$ w.r.t. a range predicate $B$.
The reconstructed neighbor list $\mathcal{N}(o)$ is initialized as empty (Line~1).
Starting from the root, it descends the tree;
at each level $\ell$, it locates the unique node $p$ such that $o \in \mo(p)$ and scans the neighbors of $o$ in $G_{p}$.
If a neighbor $o_v$ has already been marked as visited in previous expansions, it is skipped (Line~5); otherwise, $o_v$ is marked as visited and appended to $\mathcal{N}(o)$ if it satisfies $B$ (Lines~6-8).
The procedure terminates once $|\mathcal{N}(o)| = c_n$ or the leaf level is reached, and returns the neighbor list $\mathcal{N}(o)$.

This reconstruction strategy is deliberately simple: more sophisticated variants would add complexity, and their benefits tend to diminish as the number of attributes $m$ increases.
In contrast, aggregating in-range neighbors along the path of $o$ provides a direct and effective way to exploit the multi-level structure of \kh.

\begin{algorithm}[t]
\caption{$\algqkhi(\mathcal{I}, Q, k, \ef, c_e, c_n)$}
\KwIn{Index \kh $\mathcal{I}$, query $Q = (q, B)$, target size $k$, exploration factor \ef, entry-point budget $c_{e}$, and neighbor budget $c_n$}
\KwOut{An object set $\hat{R}$}
\label{alg:khi_q}
$\hat{R},\mathcal{C}_q,\mathsf{visited}\gets \emptyset$\;
$\mathcal{P}\gets\kw{RangeFilter}(\mathcal{I}.T, B, c_e)$\;
\ForEach{$o\in \mathcal{P}$}{ 
    $dist\gets\delta(o.x, q)$\;
    $\hat{R}.\textsf{push}(o, dist); \mathcal{C}_q.\textsf{push}(o, dist)$\;
    $\mathsf{visited}[o]\gets \kw{true}$\;
}
\While{$\mathcal{C}_q\neq\emptyset\land (|\hat{R}|<\ef\lor \mathcal{C}_q.\textsf{top}().dist\le\hat{R}.\textsf{top}().dist)$}{
    $o_u\gets \mathcal{C}_q.\textsf{pop}()$\;
    $\mathcal{N}(o_u)\gets \kw{ReconsNbr}(\mathcal{I}, o_u, B, c_n, \mathsf{visited})$\;
    \ForEach{$o_v\in \mathcal{N}(o_u)$}{
        $dist\gets\delta(o_v.x, q)$\;
        $\hat{R}.\textsf{push}(o_v, dist); \mathcal{C}_q.\textsf{push}(o_v, dist)$\;
        \lIf{$|\hat{R}|>\ef$}{$\hat{R}.\textsf{pop}()$}
    }
}
\Return{the $k$ closest objects in $\hat{R}$}\;
\end{algorithm}

\stitle{Query algorithm.}
\refalg{khi_q} summarizes the query procedure based on \kh $\mathcal{I}$.
Given a query $Q=(q,B)$ with target size $k$, exploration factor \ef, entry-point budget $c_e$, and neighbor budget $c_n$, the algorithm maintains two priority queues: a candidate queue $\mathcal{C}_q$ and a result queue $\hat{R}$ (Line~1).
Both queues store pairs $(o,\delta(o.x,q))$, ordered by distance to $q$;
$\mathcal{C}_q$ is a min-heap, whereas $\hat{R}$ is a bounded max-heap of capacity \ef.

The algorithm first invokes~\refalg{range_filter} to obtain a set $\mathcal{P}$ of entry points in $O_B$ (Line 2).
For each $o\in \mathcal{P}$, it computes $\delta(o.x, q)$, inserts $(o,\textit{dist})$ into both $\hat{R}$ and $\mathcal{C}_q$, and marks $o$ as visited (Lines~3–6).
While $\mathcal{C}_q$ is non-empty and either $|\hat{R}| < \ef$ or the best candidate in $\mathcal{C}_q$ is no farther than the worst object in $\hat{R}$, the algorithm extracts from $\mathcal{C}_q$ the current closest object $o_u$ (Line~8).
Its neighbors $\mathcal{N}(o_u)$ under $B$ are reconstructed by~\refalg{cons_nbr} (Line~9).
For each $o_v \in \mathcal{N}(o_u)$, the distance $\delta(o_v.x,q)$ is computed and $(o_v,\textit{dist})$ is inserted into both $\hat{R}$ and $\mathcal{C}_q$ (Lines~10–12).
If $|\hat{R}|$ exceeds \ef, the farthest object in $\hat{R}$ is removed to keep the result queue bounded (Line~13).
After the greedy search terminates, the $k$ closest objects in $\hat{R}$ are returned as the query result to $Q$ (Line~14).

In our implementation, we set $c_e = k$, since larger values bring only limited gains in search quality.
The neighbor budget is set to $c_n = M$, matching the maximum degree bound in the filtered HNSW graphs for the same reason.

\begin{algorithm}[t]
\caption{\kw{BuildTree}$(O,\tau,c_l)$}
\label{alg:build_tree}
\KwIn{Object set $O$, balance threshold $\tau$, leaf capacity $c_l$}
\KwOut{Partitioning tree $T$}
$p_r \gets$ \text{the root node of} $T$\;
$\mo(p_r)\gets O; \mdim(p_r)\gets 1; \mathsf{BL}(p_r)\gets \emptyset$\;
$\mathsf{stack}.\textsf{push}(p_r)$\;
\While{$\mathsf{stack}\neq\emptyset$}{
  $p \gets \mathsf{stack}.\mathsf{pop}()$\;
  \lIf{$|\mo(p)| \le c_l \lor |\mathsf{BL}(p)| = m$}{\textbf{continue}}
  \While{$\mdim(p) \in \mathsf{BL}(p)$}{
     $\mdim(p) \gets (\mdim(p) \bmod m) + 1$\;
  }
  sort objects in $\mo(p)$ increasingly by $o.t.t_{\mdim(p)}$\;
  set $s(p)$ as the lower median\;
  $\mo(p_l)\gets \{o\mid o.t.t_{\mdim(p)} \le s(p)\}$\;
  $\mo(p_r)\gets \mo(p)\setminus \mo(p_l)$\;
  \eIf{$\tau \cdot  \min(|\mo(p_l)|, |\mo(p_r)|) \le \max(|\mo(p_l)|, |\mo(p_r)|)$}{
    $\mathsf{BL}(p)\gets \mathsf{BL}(p)\cup \{\mdim(p)\}$\;
    $\mathsf{stack}.\mathsf{push}(p);$ \textbf{continue}\;
  }{
    $\mdim(p_l), \mdim(p_r) \gets (\mdim(p) \bmod m) + 1$\;
    $\mathsf{BL}(p_l), \mathsf{BL}(p_r)\gets \mathsf{BL}(p)$\;
    the left and right child of $p\gets (p_l, p_r)$\;
    $\mathsf{stack}.\mathsf{push}(p_l)$; $\mathsf{stack}.\mathsf{push}(p_r)$\;
  }
}
\Return{$T$}\;
\end{algorithm}

\subsection{Index Construction}
\label{subsec:idx_cons}
Considering the structure of \kh, index construction proceeds as follows.
First, we build the partitioning tree $T$ over all objects in the attribute space.
Then, leveraging the hierarchy of $T$, we construct a single-level filtered HNSW graph at each node to capture proximity relationships among the corresponding embeddings.

\stitle{Partitioning tree construction.}
We construct the partitioning tree $T$ in a top-down manner, as shown in~\refalg{build_tree}.
Starting from the root node $p_r$, each node $p$ stores a splitting dimension $\mdim(p)$ and a set of excluded dimensions $\mathsf{BL}(p)\subseteq \{1,\ldots,m\}$.
The construction proceeds through a stack-based traversal.

When a node $p$ is popped from the stack, partitioning stops if $|\mo(p)| \le c_l$ or all $m$ dimensions have been included in $\mathsf{BL}(p)$ (Line~6).
Otherwise, we advance $\mdim(p)$ in a round-robin manner until we find an eligible (non-excluded) dimension (Lines~7-8), and then determine the split value $s(p)$ on dimension $\mdim(p)$ (Lines~9-10).
Concretely, let $\{o.t.t_{\mdim(p)} \mid o \in \mo(p)\}$ be the multiset of attribute values on this dimension, sorted in ascending order.
We define $\mathsf{mid} = \lfloor (|\mo(p)| - 1)/2 \rfloor$ and choose $s(p)$ as the value at position $\mathsf{mid}$, which induces two subsets $\mo(p_l)$ and $\mo(p_r)$ for the left and right child, respectively (Lines~11–12).
We then check the balance condition.
If $\tau \cdot \min(|\mo(p_l)|,|\mo(p_r)|) \le \max(|\mo(p_l)|,|\mo(p_r)|)$, the split is considered skewed:
$\mdim(p)$ is added to $\mathsf{BL}(p)$ and $p$ is pushed back onto the stack to retry with another dimension (Lines~14–15).
Otherwise, we accept the split, create the two children, assign to each the next splitting dimension in the round-robin order, inherit the set of excluded dimensions, and push them onto the stack to continue partitioning (Lines~17–20).

\begin{algorithm}[t!]
\caption{$\kw{BuildGraph}(T, M, \ef_b)$}
\KwIn{Partitioning tree $T$, maximum degree bound $M$, and build exploration factor $\ef_b$}
\KwOut{Filtered HNSW graphs in \kh}
\label{alg:khi_cons}
\tcp{tree nodes ordered from leaves to root}
$\kw{queue}\gets\kw{BottomUpLevelTraversal}(T)$\;
\While{$\kw{queue}\neq \emptyset$}{
    $p\gets \kw{queue}.\kw{pop}()$\;
    \If{$p$ is a leaf} {
        build the filtered HNSW graph $G_{p}$\;
    } \Else {
        $(p_l, p_r)\gets$ the left and right child of $p$\;
        $G_{p}\gets G_{p_l}$\;
        \ForEach{$o\in \mo(p_r)$}{
            $\hat{R}\gets\alggs(o, \ef_b, G_{p})$\;
            $\mathcal{N}(o)$ in $G_{p}\gets \kw{Prune}(\hat{R}\cup\mathcal{N}(o)$ in $G_{p_r})$\;
            \ForEach{$\hat{o}\in \mo(p_l)\cap\mathcal{N}(o)$ in $G_{p}$}{
                $\mathcal{N}(\hat{o})$ in $G_{p}\gets\kw{Prune}(\{o\}\cup\mathcal{N}(\hat{o})$ in $G_{p})$\;
            }
        }
    }
}
\Return{$\{G_{p}\}$}\;
\end{algorithm}
\stitle{Filtered HNSW graph construction.}
Given the partitioning tree $T$, we construct for each node $p$ a single-level filtered HNSW graph $G_{p}$ in a bottom-up manner (\refalg{khi_cons}).
We first build a queue of tree nodes ordered from leaves to root: nodes are grouped by depth, and the levels are enqueued from the deepest level up to the root, ensuring that every child is processed before its parent (Line~1).

If the dequeued node $p$ is a leaf, we build $G_{p}$ by invoking the standard HNSW building procedure with maximum degree bound $M$ (Lines~4–5).
Otherwise, we set $G_{p}$ to $G_{p_l}$ (Line~8), and then incrementally merge the objects in $\mo(p_r)$ into $G_{p}$.
For each object $o \in \mo(p_r)$, we run a greedy search on the current $G_{p}$ to obtain a candidate set $\hat{R}$ of up to $\ef_b$ objects (Line~10).
Its neighbor list in $G_{p}$ is then determined by applying the RNG-based heuristic from~\refsubsec{graph} to the union of candidates from the parent and right-child graphs, obtaining an $M$-bounded neighbor list while retaining high-quality neighbors (Line~11).
Finally, we update neighbor lists of affected objects in $\mo(p_l)$: for each $\hat{o} \in \mo(p_l) \cap \mathcal{N}(o)$ in $G_{p}$, we refine $\mathcal{N}(\hat{o})$ in $G_{p}$ by applying the same RNG-based pruning to ${o} \cup \mathcal{N}(\hat{o})$ in $G_{p}$ (Lines~12–13).
Processing all nodes in this bottom-up order yields the collection $\{G_{p}\}$ of graphs attached to the nodes of \kh (Line~14).
In our implementation, we set $\ef_b=M$.

\stitle{Parallel graph construction.}
In the overall index construction, building the filtered HNSW graphs dominates the runtime because it requires many distance computations between high-dimensional embeddings (see~\cite{khi-tech-report} for details).
To reduce construction time, we parallelize graph construction while keeping the resulting graphs identical to the sequential version.

Recall that the nodes of $T$ are grouped by depth in a bottom-up order.
Since the construction of graph $G_{p}$ at a node $p$ only depends on the graphs of its children, all nodes on the same level can be processed independently, so we first build graphs in a level-wise parallel fashion.
However, near the upper levels of the tree, the number of nodes per level becomes small and the workload across nodes can become imbalanced, limiting the benefit of purely level-wise parallelization.
To address this, we introduce a threshold $\tau_p$ on the number of nodes per level: when a level contains fewer than $\tau_p$ nodes, we switch from level-wise parallelism to intra-node parallelism.
In this case, each remaining node $p$ typically has a large $\mo(p_r)$, and we parallelize the merge step for objects in $\mo(p_r)$ (Lines~9–13 of~\refalg{khi_cons}) by distributing these objects across worker threads, each independently running the greedy search and RNG-based pruning to update the corresponding neighbor lists in $G_{p}$.
In our implementation, we set $\tau_p = 100$.

\begin{table*}[t!]
\centering
\caption{Datasets.}
\vspace{-1em}
\label{tab:data}
\begin{tabular}{l|l|l|l|l|l}
    \hline
    Dataset & $n$ & $d$ & $m$ & Vector Type & Attributes \\ \hline \hline
    Youtube & 3,650,716 & 1,024 & 4 & Video & PublishYear, \#Views, \#Likes, \#Comments \\ \hline
    DBLP & 6,357,867 & 768 & 4 & Text & PublishYear, \#Citations, \#References, \#Authors  \\ \hline
    MSMarco & 8,000,000 & 384 & 5 & Text & \#Words, \#Chars, \#Sentences, \#UniqueWords, TFIDF  \\ \hline
    LAION & 9,636,707 & 512 & 3 & Image & Width, Height, Similarity  \\ \hline
\end{tabular}
\vspace{-1em}
\end{table*}

\stitle{Discussion.}
Our partitioning tree $T$ is inspired by classical KD-trees but departs from standard variants in several key aspects.
Each node maintains a set of excluded dimensions $\mathsf{BL}(p)$ and uses a balance threshold $\tau$ to avoid highly skewed splits, which yields a provable bound on the tree height.
Moreover, $T$ is tailored to RFANNS: it is built purely in the attribute space, and its nodes serve as anchors for filtered HNSW graphs rather than for traditional point-location queries as in standard KD-trees.
These design choices are driven by the requirements of RFANNS and, to the best of our knowledge, have not been explored in existing KD-tree–based indexes.

\section{Experiments}
\label{sec:experiment}
\begin{figure*}[htb]
\centering
    \vspace{-1em}
    \begin{minipage}{0.59\textwidth}
        \centering
        \includegraphics[width=0.99\hsize]{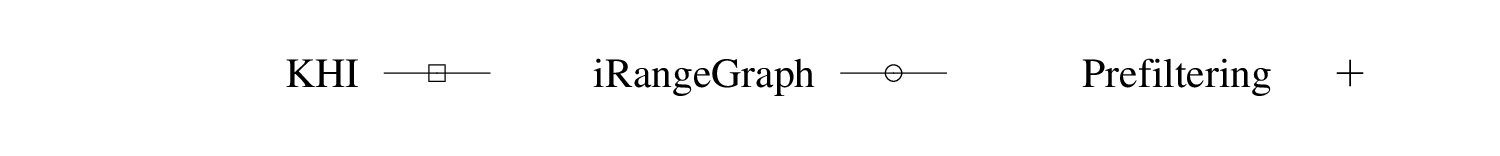}
    \end{minipage}\\
    \vspace{-1.5em}
    \begin{subfigure}[b]{0.2\textwidth}
        \centering
        \includegraphics[width=0.99\hsize]{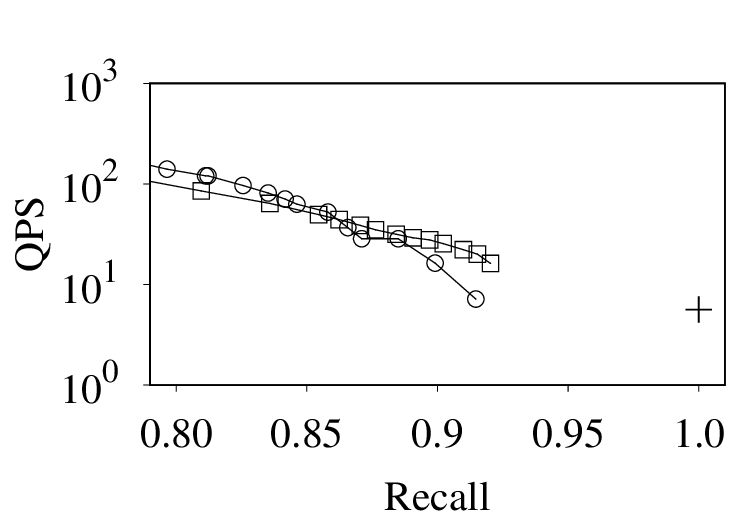}
        \vspace{-2em}
        \caption{\yt ($\sigma$=1/16)}
    \end{subfigure}
    \begin{subfigure}[b]{0.2\textwidth}
        \centering
        \includegraphics[width=0.99\hsize]{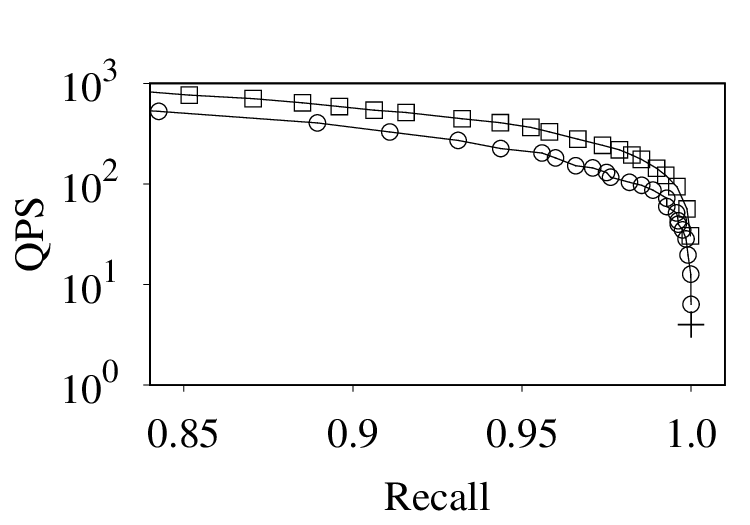}
        \vspace{-2em}
       \caption{\dblp ($\sigma$=1/16)}
    \end{subfigure}
    \begin{subfigure}[b]{0.2\textwidth}
        \centering
        \includegraphics[width=0.99\hsize]{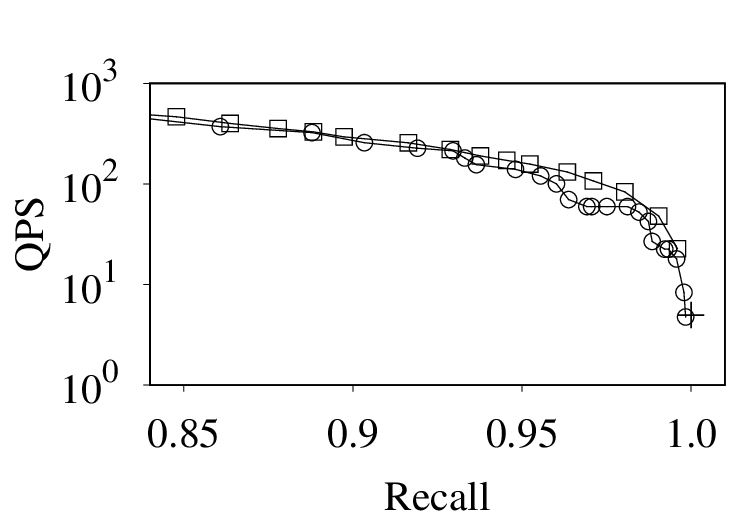}
        \vspace{-2em}
        \caption{\ms ($\sigma$=1/16)}
    \end{subfigure}
    \begin{subfigure}[b]{0.2\textwidth}
        \centering
        \includegraphics[width=0.99\hsize]{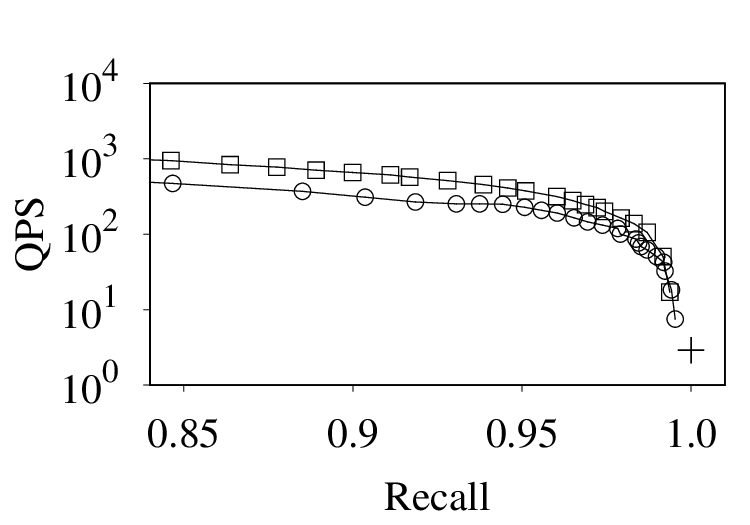}
        \vspace{-2em}
       \caption{\laion ($\sigma$=1/16)}
    \end{subfigure}\\
    \begin{subfigure}[b]{0.2\textwidth}
        \centering
        \includegraphics[width=0.99\hsize]{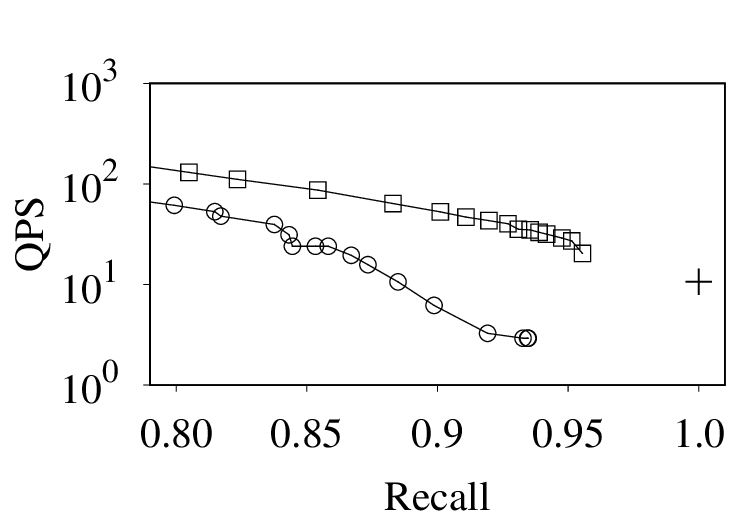}
        \vspace{-2em}
        \caption{\yt ($\sigma$=1/64)}
        \label{fig:exp:overall_yt_64}
    \end{subfigure}
    \begin{subfigure}[b]{0.2\textwidth}
        \centering
        \includegraphics[width=0.99\hsize]{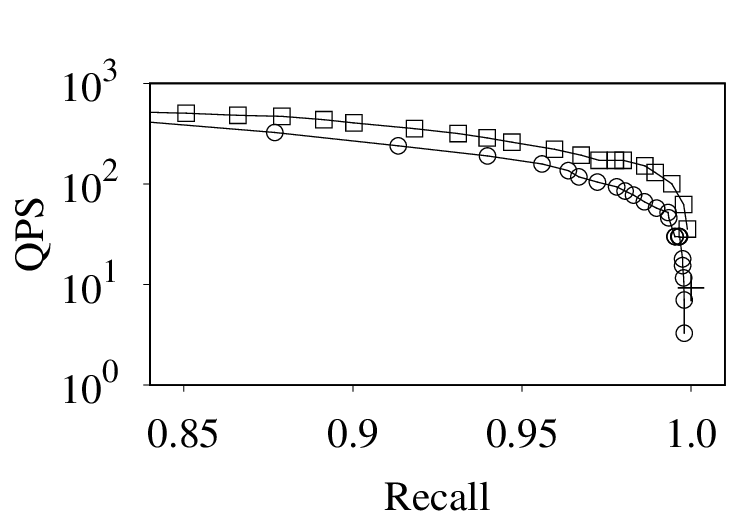}
        \vspace{-2em}
       \caption{\dblp ($\sigma$=1/64)}
    \end{subfigure}
    \begin{subfigure}[b]{0.2\textwidth}
        \centering
        \includegraphics[width=0.99\hsize]{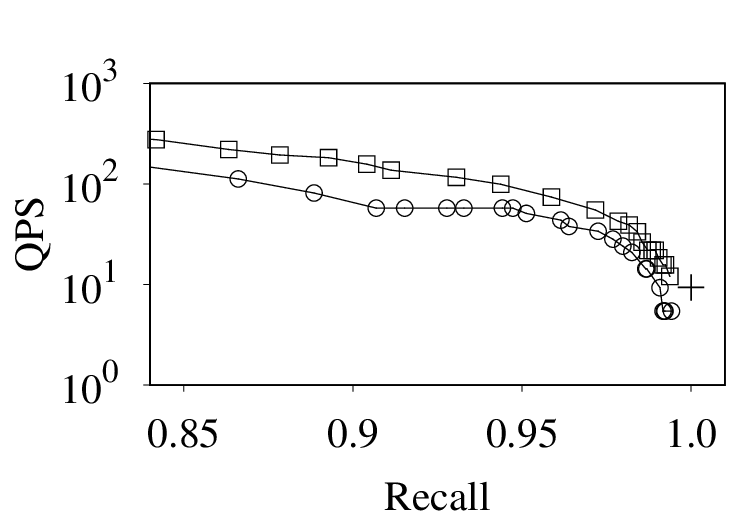}
        \vspace{-2em}
        \caption{\ms ($\sigma$=1/64)}
    \end{subfigure}
    \begin{subfigure}[b]{0.2\textwidth}
        \centering
        \includegraphics[width=0.99\hsize]{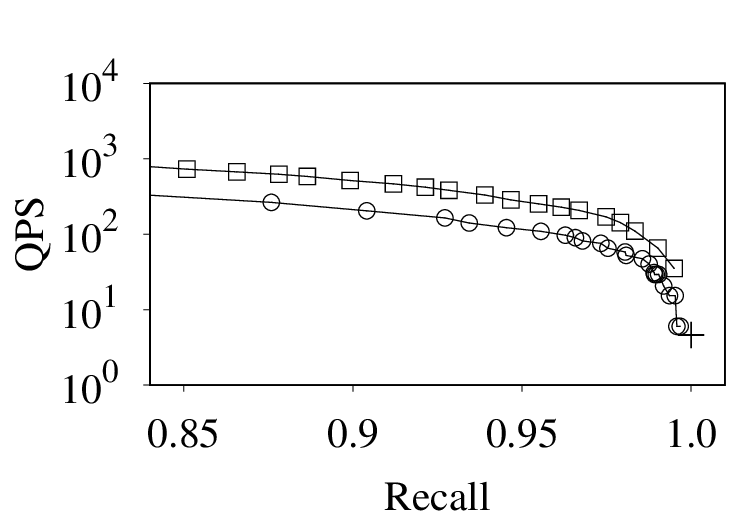}
        \vspace{-2em}
       \caption{\laion ($\sigma$=1/64)}
    \end{subfigure}\\
    \begin{subfigure}[b]{0.2\textwidth}
        \centering
        \includegraphics[width=0.99\hsize]{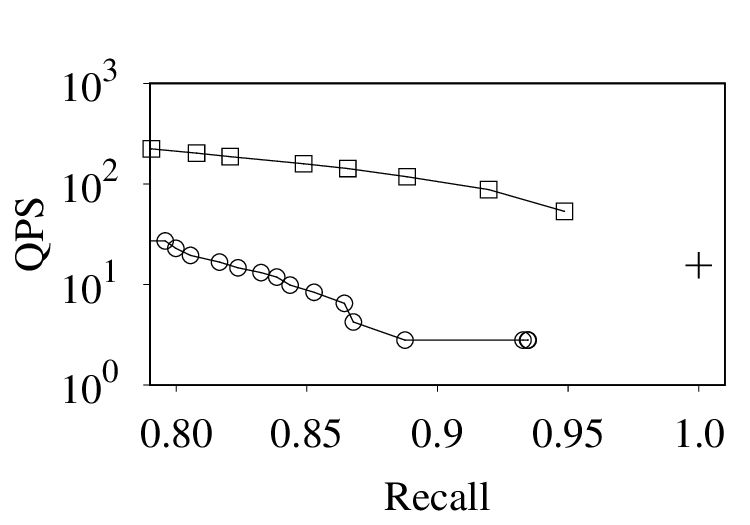}
        \vspace{-2em}
        \caption{\yt ($\sigma$=1/256)}
        \label{fig:exp:overall_yt_256}
    \end{subfigure}
    \begin{subfigure}[b]{0.2\textwidth}
        \centering
        \includegraphics[width=0.99\hsize]{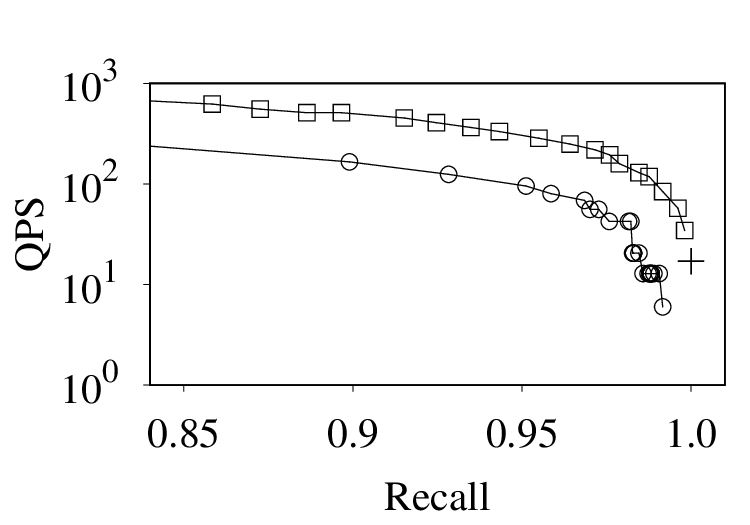}
        \vspace{-2em}
       \caption{\dblp ($\sigma$=1/256)}
    \end{subfigure}
    \begin{subfigure}[b]{0.2\textwidth}
        \centering
        \includegraphics[width=0.99\hsize]{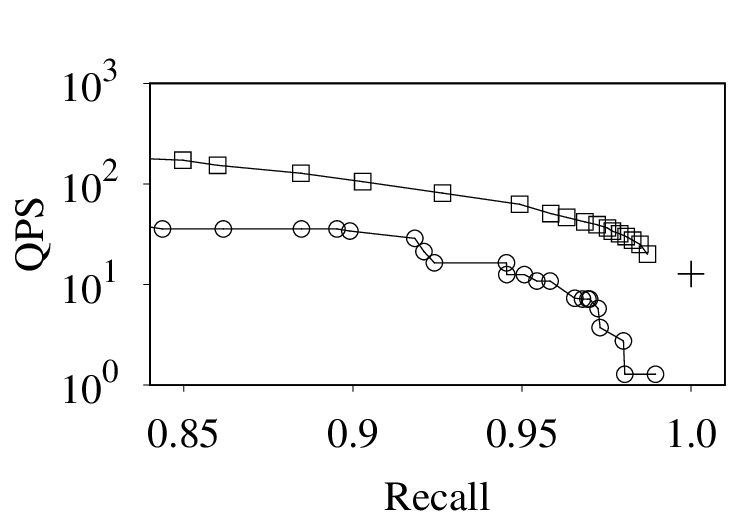}
        \vspace{-2em}
        \caption{\ms ($\sigma$=1/256)}
    \end{subfigure}
    \begin{subfigure}[b]{0.2\textwidth}
        \centering
        \includegraphics[width=0.99\hsize]{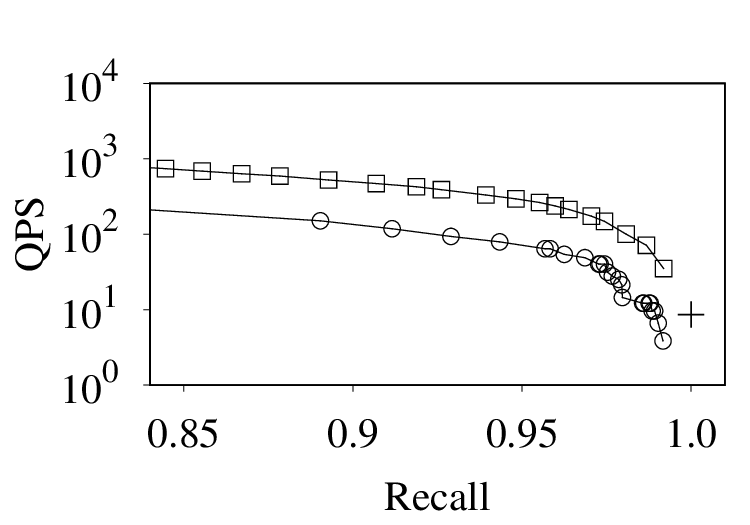}
        \vspace{-2em}
       \caption{\laion ($\sigma$=1/256)}
    \end{subfigure}
    \vspace{-1em}
    \caption{Overall query performance.}
    \label{fig:exp:overall_q}
    \Description*{}
    \vspace{-1em}
\end{figure*}

\subsection{Experimental Setup}
\label{subsec:setup}
Our experimental evaluation focuses on two aspects: query performance and index efficiency.
For query performance, we compare our method against baselines under various query settings.
For index efficiency, we report the index construction time and the empirical space usage of our method and baselines.

\stitle{Datasets.}
To evaluate the performance of all methods, we use four real-world public datasets.
Since few public corpora provide both multi-attribute numeric metadata and corresponding vector representations, we construct our evaluation datasets by encoding the raw content with open-source embedding models and pairing the resulting vectors with the original attributes.
The resulting datasets are summarized in~\reftab{data}.
\laion\footnote{\url{https://laion.ai/blog/laion-400-open-dataset/}} contains image–text pairs with three numeric attributes: image width, image height, and a similarity score.
\ms\footnote{\url{https://microsoft.github.io/msmarco/}} and \dblp\footnote{\url{https://open.aminer.cn/open/article?id=655db2202ab17a072284bc0c}} are text collections with document-level statistics such as the number of words, sentences, and citations.
\yt\footnote{\url{https://research.google.com/youtube8m/}} consists of video records with temporal and popularity metadata, including publication year, the number of views, and the number of likes.

\stitle{Algorithms.}
Among existing RFANNS indexes, we consider two representative methods, SeRF and iRangeGraph, as introduced in~\refsubsec{rfann}.
SeRF is designed for single-attribute range filters. Although its paper outlines a possible extension to multi-attribute RFANNS, this variant is only described at a high level and, to the best of our knowledge, has not been implemented.
Since our focus is on multi-attribute RFANNS, we do not evaluate SeRF experimentally and instead use iRangeGraph as the main RFANNS baseline.

In addition to iRangeGraph, we also include a simple pre-filtering baseline.
Given a query $Q=(q,B)$, \kw{Prefiltering} first scans all objects to materialize the subset $O_B$.
It then performs exact nearest neighbor search for $q$ over $O_B$ in the embedding space by exhaustively computing distances and returning the $k$ closest objects.

\stitle{Queries.}
For each dataset, we generate $1{,}000$ RFANNS queries of the form $Q = (q,B)$.
During dataset construction, we randomly sample $1{,}000$ raw objects from the original corpus and encode them with the same embedding model used to build the dataset; these vectors are stored separately and used only as query embeddings.

The range predicate $B$ is generated in the attribute space with a target selectivity $\sigma$ and a relative tolerance parameter $\mathsf{tol}$.
Following iRangeGraph, we parameterize the target selectivity as $\sigma = 1/2^i$.
Since multi-attribute range predicates often result in relatively small selectivities in practice, we focus on $i \in \{4,6,8\}$, i.e., $\sigma \in \{1/16, 1/64, 1/256\}$.

To instantiate $B$, we first draw a random sample of attribute tuples and, for each attribute, build a sorted list of finite values to support quantile queries.
For each query, we derive per-attribute intervals by selecting lower and upper quantiles on the sampled tuples so that the empirical selectivity of the resulting $B$ lies within $[\sigma(1-\mathsf{tol}),\,\sigma(1+\mathsf{tol})]$.
By default, we set $\mathsf{tol}=0.5$.

\stitle{Equipment.}
All algorithms are implemented in C++, compiled with the \texttt{g++} compiler at \texttt{-O3} optimization level, and run on a Linux machine equipped with an Intel Xeon Gold 6230 CPU @ 2.10\,GHz and 256\,GB of RAM.

\subsection{Query performance}
\label{subsec:exp:query}
\stitle{Overall performance.}
We evaluate query efficiency using queries per second (QPS), defined as the number of RFANNS queries processed per second, and report the trade-off between recall and QPS.
\reffig{exp:overall_q} summarizes the query performance of the three methods on all four datasets.
For each figure, the curve is obtained by varying the exploration factor $\ef$ on a fixed index, yielding a series of points with different recall–QPS combinations.
Both \kh and \rg use a maximum degree bound $M=32$.
In this experiment, we set the cardinality of $B$ to $m$ (i.e., each range predicate constrains all attributes) and fix the target size $k=10$.
We make the following observations:

(1) The two graph-based methods, \kh and \rg, achieve high recall across all datasets and selectivities: on \laion, \ms, and \dblp the recall exceeds $0.99$, and on the more challenging \yt\ it remains above $0.9$.
However, their query throughput shows clear differences.

(2) When $\sigma \in \{1/16,1/64,1/256\}$, \kh\ consistently outperforms \rg\ in QPS at comparable recall levels.
At recall $0.95$, the QPS gains on \laion\ are $1.66\times$, $2.34\times$, and $3.99\times$ over \rg\ as $\sigma$ decreases from $1/16$ to $1/256$; 
on \ms, the corresponding speedups are $1.20\times$, $1.67\times$, and $4.92\times$; 
and on \dblp, they are $1.76\times$, $1.48\times$, and $3.15\times$, respectively.
On the more challenging \yt\ dataset, \kh\ yields $1.68\times$, $8.89\times$, and $38.08\times$ speedups at recall $0.9$ for $\sigma=1/16$, $1/64$, and $1/256$, respectively.
Averaged over \laion, \ms, and \dblp\ and the three selectivities, \kh\ attains an overall speedup of $2.46\times$ in QPS compared with \rg, while on \yt\ the average speedup reaches $16.22\times$.
Overall, the advantage grows more pronounced for lower selectivities and more challenging datasets.

(3) When the target recall is relaxed below $1.0$, both \kh and \rg yield notable speedups over \kw{Prefiltering}.
At recall $0.95$ on \laion, \ms, and \dblp\ (and $0.9$ on \yt), \kh consistently achieves higher QPS than \kw{Prefiltering} on all datasets, with an average speedup of $35.59\times$.
In the same setting, \rg also outperforms \kw{Prefiltering} on all datasets except \yt, with an average speedup of $24.91\times$; 
on \yt, however, \rg falls short of \kw{Prefiltering}, as shown in~\reffig{exp:overall_yt_64} and~\reffig{exp:overall_yt_256}.

\begin{figure}[t!]
\centering
    \begin{minipage}{0.35\textwidth}
        \centering
        \includegraphics[width=0.99\hsize]{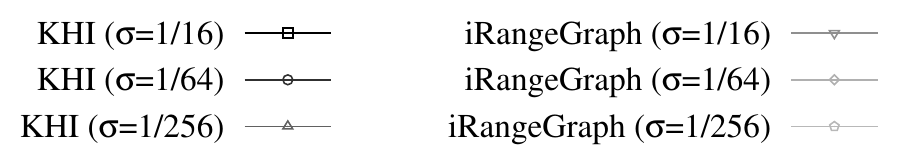}
    \end{minipage}\\
    \begin{subfigure}[b]{0.2\textwidth}
        \centering
        \includegraphics[width=0.99\hsize]{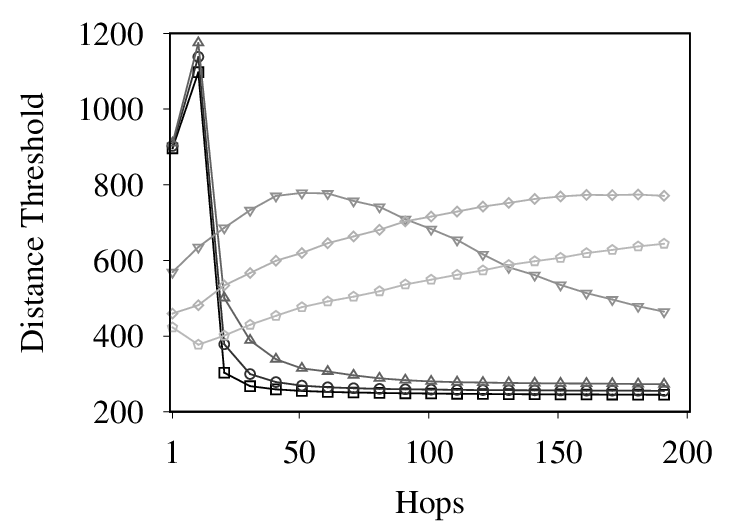}
        \vspace{-2em}
        \caption{\yt}
    \end{subfigure}
    \begin{subfigure}[b]{0.2\textwidth}
        \centering
        \includegraphics[width=0.99\hsize]{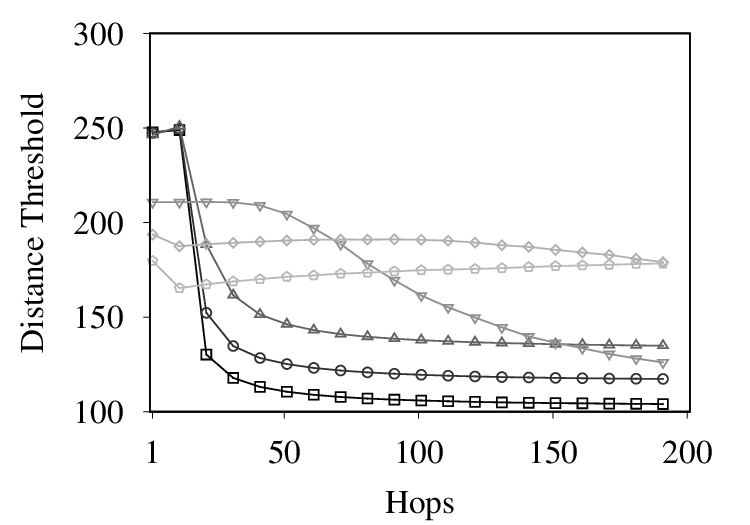}
        \vspace{-2em}
        \caption{\dblp}
    \end{subfigure}\\
    \begin{subfigure}[b]{0.2\textwidth}
        \centering
        \includegraphics[width=0.99\hsize]{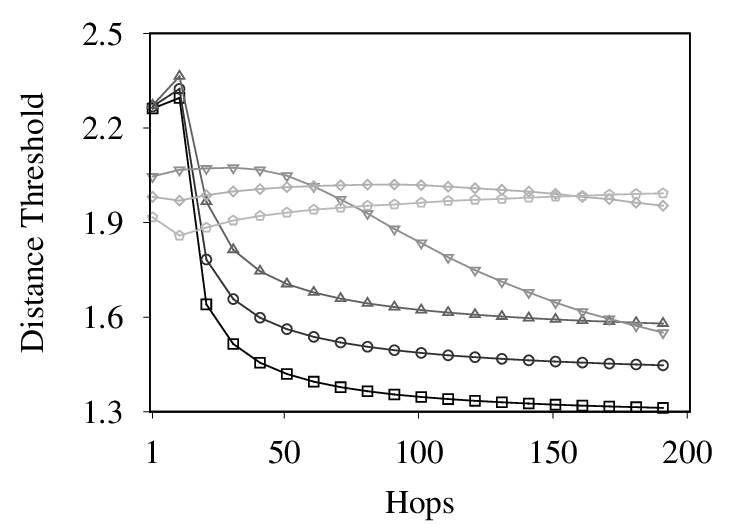}
        \vspace{-2em}
        \caption{\ms}
    \end{subfigure}
    \begin{subfigure}[b]{0.2\textwidth}
        \centering
        \includegraphics[width=0.99\hsize]{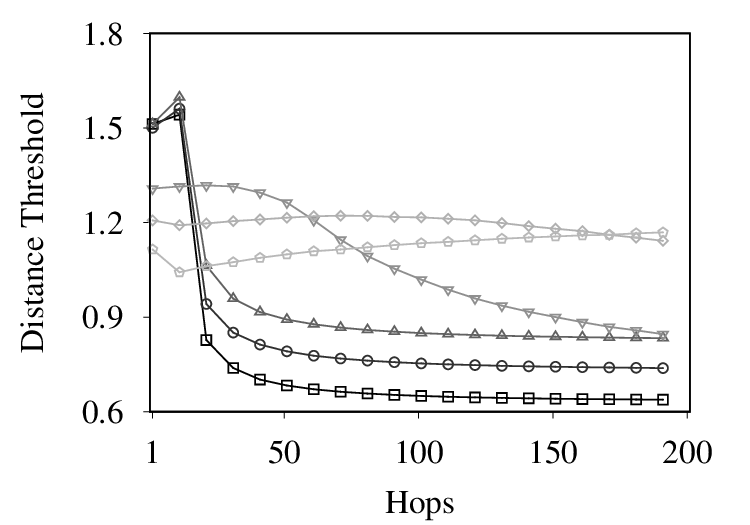}
        \vspace{-2em}
        \caption{\laion}
    \end{subfigure}
    \vspace{-1em}
    \caption{Evolution of distance threshold during search.}
    \Description*{}
    \label{fig:exp:dt}
\end{figure}

\stitle{Convergence of the distance threshold.}
To better understand how the graph-based methods behave, we further examine the evolution of the distance threshold during greedy search.
We measure the search progress in $\kw{hops}$, where one hop corresponds to expanding the neighbors of a candidate object.
At a given hop, we define the distance threshold as the distance from $q$ to the farthest object in $\hat{R}$, where $\hat{R}$ denotes the current set of best-so-far candidates.
In this experiment, we use the same settings as in the overall query performance study.
We report the evolution of the distance threshold at $\ef=300$ for $\sigma \in \{1/16,1/64,1/256\}$.

\reffig{exp:dt} shows how this threshold evolves with the number of hops across all datasets and selectivities.
For \kh, across all datasets and selectivities, the distance threshold decreases rapidly in the first few hops and then quickly stabilizes.
This suggests that \kh\ is able to tighten the distance threshold more aggressively and prune inferior candidates earlier, while still preserving high recall.
In contrast, for \rg, the distance threshold decays much more slowly and often remains high over many hops, especially for smaller selectivities, e.g., $\sigma = 1/64$ and $\sigma = 1/256$.
As discussed in~\refsubsec{motivation}, this behavior stems from its reliance on many out-of-range neighbors.
Thus, \rg\ tends to explore more candidates during search.
These search dynamics are consistent with the higher QPS achieved by \kh, compared with \rg\ at comparable recall levels in~\reffig{exp:overall_q}.

\stitle{Varying target size $k$.}
We study the impact of the target size $k$ on query performance and present the results on \laion as a representative example.
We fix the target recall at $0.95$ and evaluate QPS while varying $k$ from $20$ to $100$, as shown in~\reffig{exp:vary_k}; the case $k=10$ has already been reported in~\reffig{exp:overall_q}.
All other settings follow the overall query performance study.

\begin{figure*}[htb]
\centering
    \vspace{-2em}
    \begin{minipage}{0.59\textwidth}
        \centering
        \includegraphics[width=0.99\hsize]{experiments/query_legend.eps}
    \end{minipage}\\
    \vspace{-1.5em}
    \begin{subfigure}[b]{0.2\textwidth}
        \centering
        \includegraphics[width=0.99\hsize]{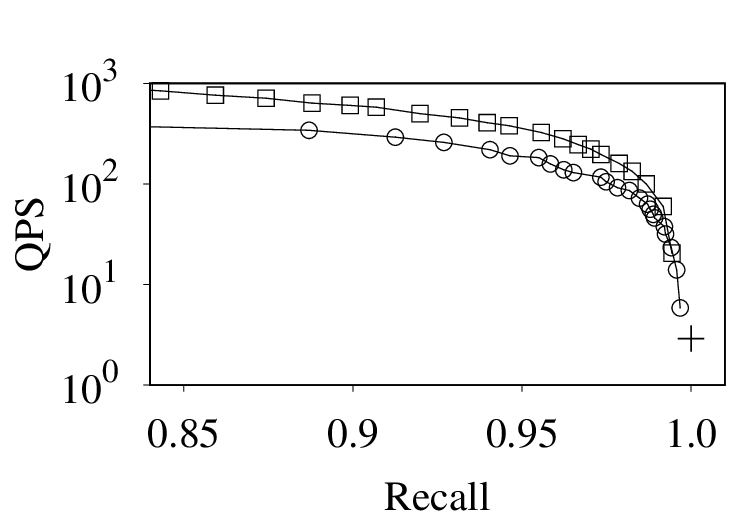}
        \vspace{-2em}
        \caption{$\sigma$=1/16, k=20}
    \end{subfigure}
    \hspace{0.08\textwidth}
    \begin{subfigure}[b]{0.2\textwidth}
        \centering
        \includegraphics[width=0.99\hsize]{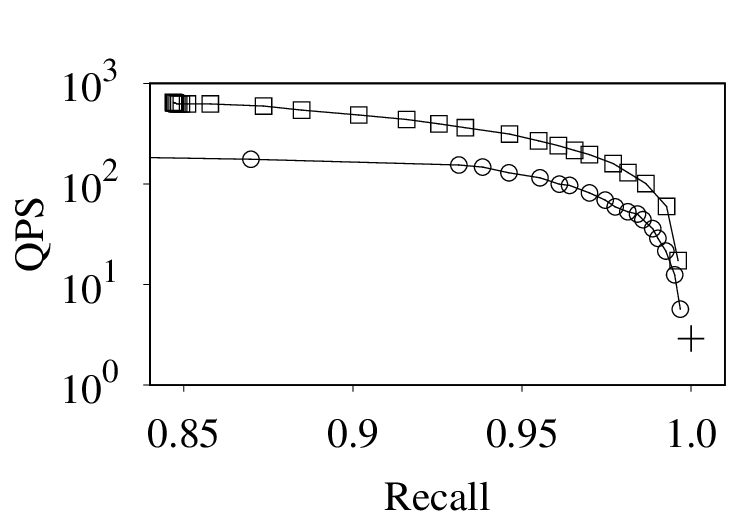}
        \vspace{-2em}
        \caption{$\sigma$=1/16, k=50}
    \end{subfigure}
    \hspace{0.08\textwidth}
    \begin{subfigure}[b]{0.2\textwidth}
        \centering
        \includegraphics[width=0.99\hsize]{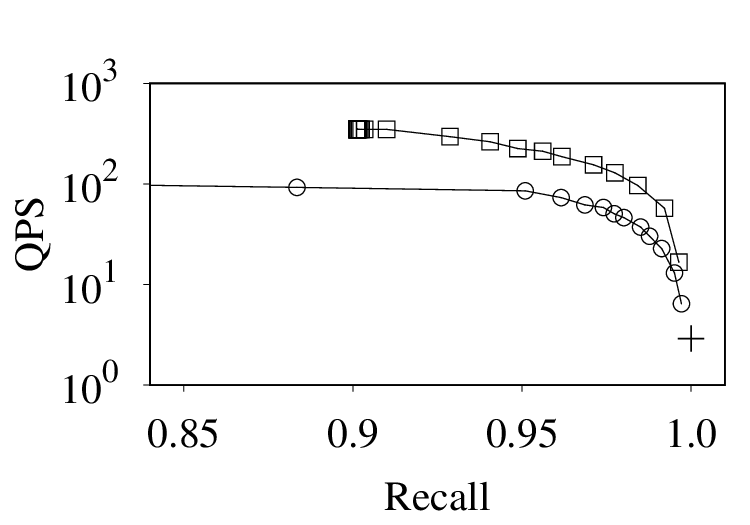}
        \vspace{-2em}
        \caption{$\sigma$=1/16, k=100}
    \end{subfigure}\\
    \begin{subfigure}[b]{0.2\textwidth}
        \centering
        \includegraphics[width=0.99\hsize]{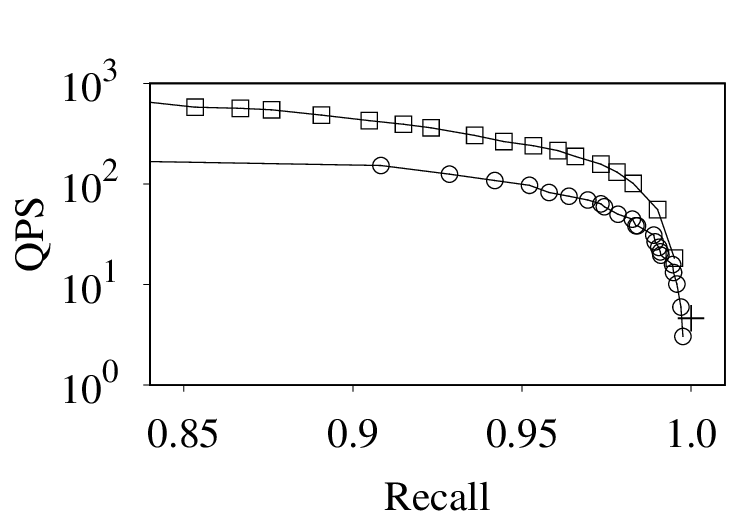}
        \vspace{-2em}
        \caption{$\sigma$=1/64, k=20}
    \end{subfigure}
    \hspace{0.08\textwidth}
    \begin{subfigure}[b]{0.2\textwidth}
        \centering
        \includegraphics[width=0.99\hsize]{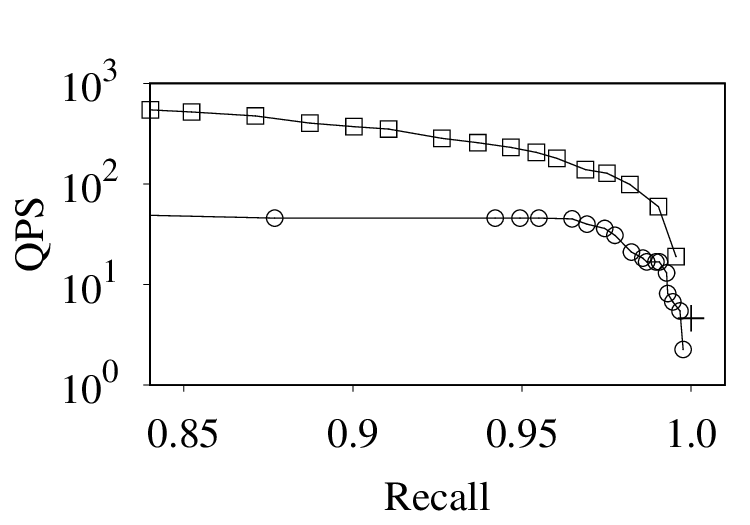}
        \vspace{-2em}
        \caption{$\sigma$=1/64, k=50}
    \end{subfigure}
    \hspace{0.08\textwidth}
    \begin{subfigure}[b]{0.2\textwidth}
        \centering
        \includegraphics[width=0.99\hsize]{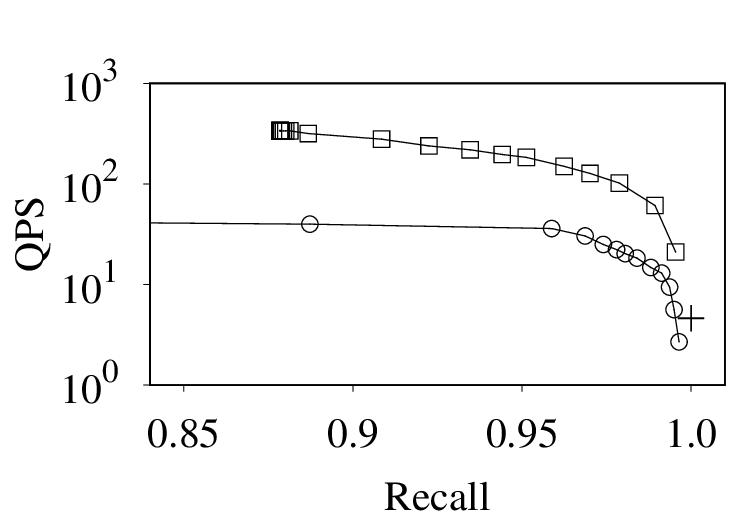}
        \vspace{-2em}
        \caption{$\sigma$=1/64, k=100}
    \end{subfigure}\\
    \begin{subfigure}[b]{0.2\textwidth}
        \centering
        \includegraphics[width=0.99\hsize]{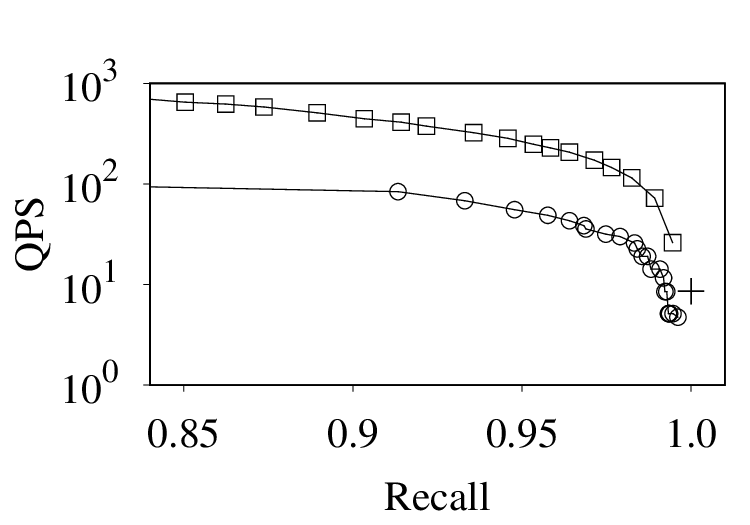}
        \vspace{-2em}
       \caption{$\sigma$=1/256, k=20}
    \end{subfigure}
    \hspace{0.08\textwidth}
    \begin{subfigure}[b]{0.2\textwidth}
        \centering
        \includegraphics[width=0.99\hsize]{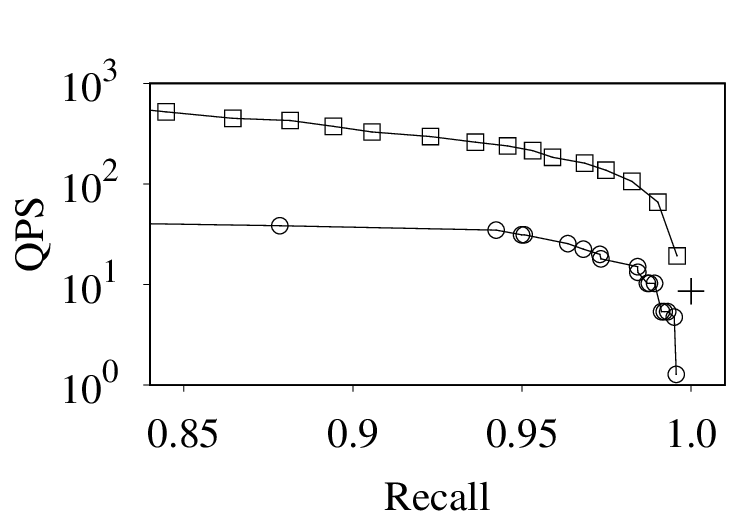}
        \vspace{-2em}
       \caption{$\sigma$=1/256, k=50}
    \end{subfigure}
    \hspace{0.08\textwidth}
    \begin{subfigure}[b]{0.2\textwidth}
        \centering
        \includegraphics[width=0.99\hsize]{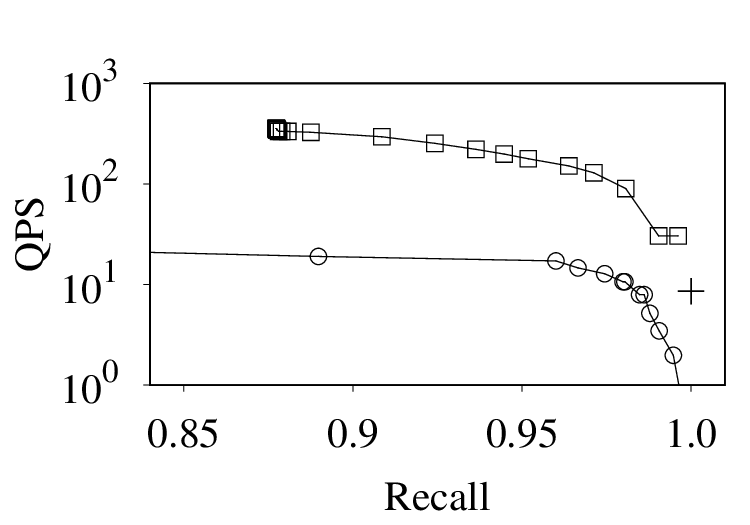}
        \vspace{-2em}
       \caption{$\sigma$=1/256, k=100}
    \end{subfigure}
    \vspace{-1em}
    \caption{Query performance by varying target size $k$.}
    \Description*{}
    \label{fig:exp:vary_k}
    \vspace{-1em}
\end{figure*}

We fix the target recall at $0.95$ and report the QPS achieved by \kh\ and \rg\ on \laion\ while varying $k \in \{10,20,50,100\}$.
When $\sigma = 1/16$, \kh\ achieves speedups over \rg, with QPS improvements of $1.66\times$, $1.91\times$, $2.38\times$, and $2.60\times$ at $k=10$, $20$, $50$, and $100$, respectively.
For $\sigma = 1/64$, the gains become larger, with speedups of about $2.34\times$, $2.51\times$, $4.81\times$, and $5.10\times$ for $k=10$, $20$, $50$, and $100$, respectively.
For the case $\sigma = 1/256$, the advantage of \kh\ is even more pronounced, reaching $3.99\times$, $4.90\times$, $7.20\times$, and $10.52\times$ for $k=10$, $20$, $50$, and $100$, respectively.
Overall, the performance gap between \kh\ and \rg\ widens as $k$ increases, especially under low selectivity.
This effect may be attributed to the fact that larger $k$ forces the search to explore a broader local neighborhood; in this regime, \kh\ benefits from higher-quality in-range neighbors, whereas \rg\ spends more hops on out-of-range candidates, so \kh\ can obtain the $k$ nearest neighbors with fewer expansions and consequently higher QPS.
We also observe that, when the recall is relaxed to $0.95$, both \kh\ and \rg\ achieve substantially higher QPS than \kw{Prefiltering}.
On \laion, averaging over $k \in \{10,20,50,100\}$, their QPS speedups over \kw{Prefil}-\kw{tering} are $62.76\times$ and $24.16\times$, respectively.

\begin{figure}[t!]
\centering
    \vspace{-1em}
    \begin{minipage}{0.49\textwidth}
        \centering
        \includegraphics[width=0.99\hsize]{experiments/query_legend.eps}
    \end{minipage}\\
    \vspace{-2em}
    \begin{subfigure}[b]{0.2\textwidth}
        \centering
        \includegraphics[width=0.99\hsize]{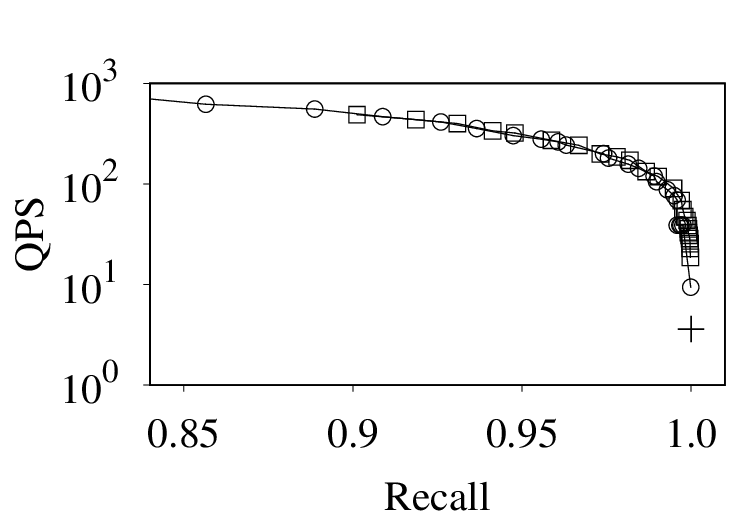}
        \vspace{-2em}
        \caption{$\sigma$=1/16, $|B|$=2}
    \end{subfigure}
    \begin{subfigure}[b]{0.2\textwidth}
        \centering
        \includegraphics[width=0.99\hsize]{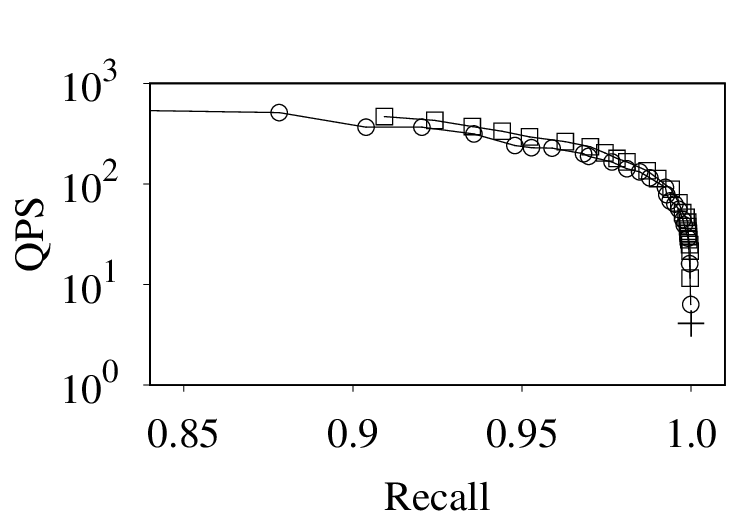}
        \vspace{-2em}
        \caption{$\sigma$=1/16, $|B|$=3}
    \end{subfigure}\\
    \begin{subfigure}[b]{0.2\textwidth}
        \centering
        \includegraphics[width=0.99\hsize]{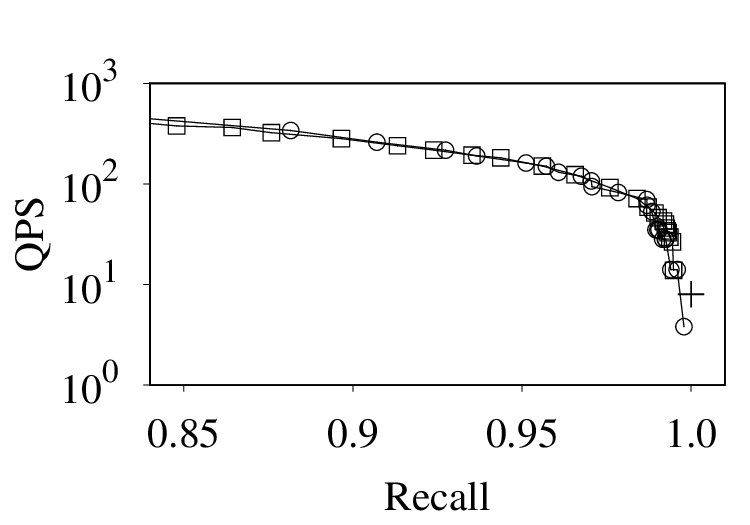}
        \vspace{-2em}
        \caption{$\sigma$=1/64, $|B|$=2}
    \end{subfigure}
    \begin{subfigure}[b]{0.2\textwidth}
        \centering
        \includegraphics[width=0.99\hsize]{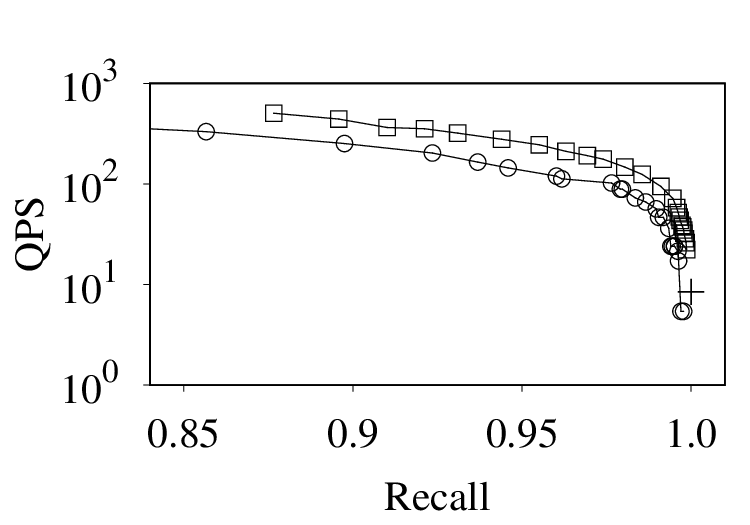}
        \vspace{-2em}
        \caption{$\sigma$=1/64, $|B|$=3}
    \end{subfigure}\\
    \begin{subfigure}[b]{0.2\textwidth}
        \centering
        \includegraphics[width=0.99\hsize]{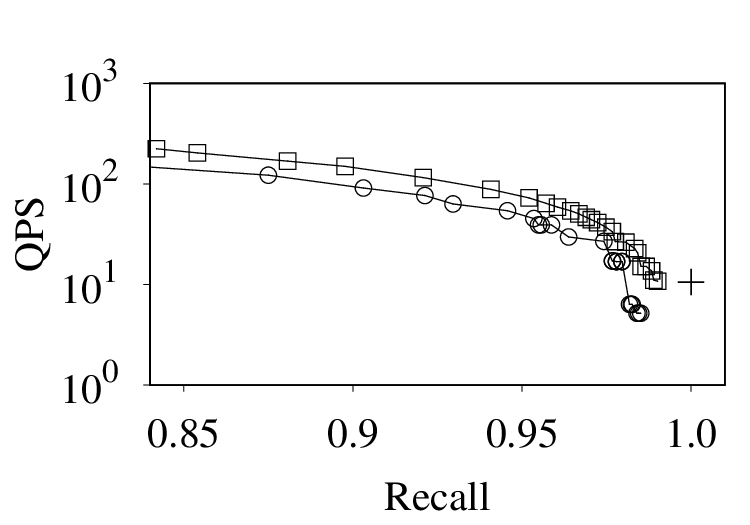}
        \vspace{-2em}
       \caption{$\sigma$=1/256, $|B|$=2}
    \end{subfigure}
    \begin{subfigure}[b]{0.2\textwidth}
        \centering
        \includegraphics[width=0.99\hsize]{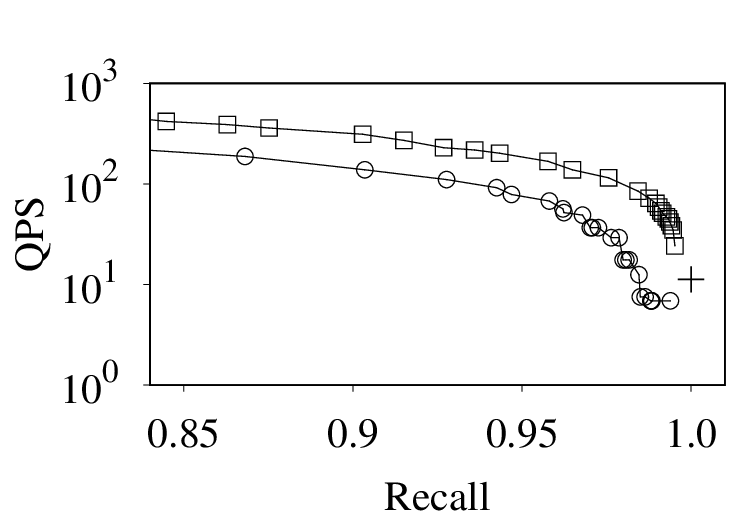}
        \vspace{-2em}
       \caption{$\sigma$=1/256, $|B|$=3}
    \end{subfigure}
    \vspace{-1em}
    \caption{Query performance by varying cardinality of $B$.}
    \Description*{}
    \label{fig:exp:vary_m}
    \vspace{-0.5em}
\end{figure}

\stitle{Varying cardinality of $B$.}
We study the effect of the cardinality of $B$ on query performance.
We report the results on \dblp; the other datasets exhibit similar trends.
We fix the target recall at $0.95$ and evaluate QPS while varying the cardinality of $B$ from $2$ to $4$.
All other settings follow the overall query performance study.
When the cardinality of $B$ is less than $m$, for each query we choose the constrained attributes uniformly at random from the $m$ attributes.
In~\reffig{exp:vary_m}, we plot the cases where the cardinality of $B$ is $2$ and $3$; the case where the cardinality of $B$ equals $m$ has already been reported in~\reffig{exp:overall_q} and is omitted here for brevity.

When $\sigma = 1/16$, \kh\ achieves about $1.06\times$, $1.29\times$, and $1.76\times$ higher QPS than \rg\ when the cardinality of $B$ is $2$, $3$, and $4$, respectively.
For $\sigma = 1/64$, the speedup ranges from roughly parity when the cardinality of $B$ is $2$ to about $1.43\times$ and $1.48\times$ when it is $3$ and $4$.
For $\sigma = 1/256$, the gains become more pronounced, reaching around $1.53\times$, $2.47\times$, and $3.15\times$ for cardinalities $2$, $3$, and $4$, respectively.
These results indicate that, as $B$ involves more attributes, \kh\ increasingly outperforms \rg, especially under low selectivity, highlighting the benefit of its attribute-space partitioning and and the resulting high-quality graphs.
Moreover, compared with the exhaustive \kw{Prefiltering} baseline on \dblp, \kh\ and \rg\ achieve average QPS speedups of $39.46\times$ and $27.81\times$, respectively, across these settings.

\subsection{Index construction}
\label{subsec:exp:idx_cons}
\begin{table}[t]
  \centering
  \caption{Index construction time.}
  \label{tab:index_construction_time}
  \vspace{-1em}
  \begin{tabular}{c|c|c}
    \hline
    Dataset & \kh (s) & \rg (s) \\
    \hline \noalign{\smallskip} \hline
    \yt    & 1,991.75  &  6,675.85 \\ \hline
    \dblp  & 4,642.37  & 13,973.42 \\ \hline
    \ms    & 5,507.91  & 17,841.86 \\ \hline
    \laion & 4,661.85 & 16,250.53 \\
    \hline
  \end{tabular}
\end{table}

\stitle{Index construction time.}
Since single-threaded index construction runs slow on our datasets, we report multi-threaded construction times under a 16-thread configuration with maximum degree bound $M=32$ for both \kh\ and \rg.
As shown in~\reftab{index_construction_time}, across the four datasets, \kh\ can be built in about $2{,}000$–$5{,}500$\,s, whereas \rg\ requires roughly $6{,}700$–$17{,}800$\,s.
These correspond to speedups of $3.35\times$, $3.01\times$, $3.24\times$, and $3.49\times$, respectively, i.e., an average reduction of construction time by about a factor of $3.27$.
This confirms that our parallelization strategy, which combines level-wise and intra-node parallelism, substantially accelerates index construction compared with the pure level-wise parallelism used in \rg.

\begin{table}[t]
  \centering
  \caption{Index size.}
  \label{tab:index_size}
  \vspace{-1em}
  \begin{tabular}{c|c|c}
    \hline
    Dataset & \kh (GB) & iRangeGraph (GB) \\
    \hline \noalign{\smallskip} \hline
    \yt    &  3.27 &  2.89 \\ \hline
    \dblp  &  9.74 & 10.14 \\ \hline
    \ms    & 18.49 & 16.89 \\ \hline
    \laion & 15.85 & 13.97 \\
    \hline
  \end{tabular}
\end{table}

\stitle{Index size.}
\reftab{index_size} reports the index sizes on all datasets, with $M=32$ for both indexes.
Overall, the two indexes occupy comparable space: across all datasets, the size of \kh\ differs from that of \rg\ by less than $15\%$.
Although our partitioning tree $T$ is not strictly balanced, the skew-aware splitting strategy keeps most objects at moderate depths, so that only a small fraction reside in deeper levels.
As a result, the additional nodes and their filtered HNSW graphs incur only limited extra space.
Further details on the empirical tree height are provided in~\cite{khi-tech-report}.

\section{Related Work}
\label{sec:related_work}

\stitle{Approximate nearest neighbor search.}
Approximate nearest neighbor search in high-dimensional Euclidean spaces has been extensively studied.
Existing methods can be broadly categorized into graph-based~\cite{malkov2014approximate,harwood2016fanng,malkov2018efficient,li2019approximate,fu2019fast,jayaram2019diskann,chen2021spann,xie2025graph}, quantization-based~\cite{jegou2010product,babenko2014additive,babenko2014inverted,andre2016cache,guo2020accelerating,gao2024rabitq}, and hashing-based~\cite{datar2004locality,tao2010efficient,gan2012locality,sun2014srs,huang2015query};
we refer readers to recent tutorials~\cite{echihabi2021new,qin2021high}, surveys~\cite{li2019approximate,wang2021comprehensive,wang2023graph}, and benchmark studies~\cite{aumuller2020ann,aumuller2023recent} for a comprehensive overview of this literature.
Among these categories, graph-based methods have demonstrated particularly strong performance:
HNSW~\cite{malkov2018efficient}, NSG~\cite{fu2019fast}, and DiskANN~\cite{jayaram2019diskann} have been widely adopted in industry~\cite{wang2021milvus}.
More recently, several GPU-accelerated graph-based ANN methods~\cite{zhao2020song,groh2022ggnn,yu2022gpu,ootomo2024cagra} have been proposed to increase throughput by exploiting massive parallelism.
These methods are highly optimized for ANN without attribute constraints.
In contrast, we focus on extending ANN to handle multi-attribute range filters efficiently, a setting not addressed in the above line of work.

\stitle{Range-filtering nearest neighbor search.}
Motivated by the practical importance of attribute-filtered ANN queries, many algorithms and systems have been proposed~\cite{cai2024navigating,engels2024approximate,gollapudi2023filtered,mohoney2023high,patel2024acorn,wang2023efficient,wei2020analyticdb,zuo2024serf,xu2024irangegraph,xie2025beyond,wang2021milvus}.
We refer readers to the survey~\cite{lin2025survey} for a comprehensive overview.
Across different application scenarios, attribute predicates can be broadly grouped into equality predicates, numeric range predicates, and more general comparison predicates.
For equality predicates~\cite{cai2024navigating,gollapudi2023filtered,wang2023efficient}, objects carry categorical labels and the goal is to retrieve the $k$ nearest neighbors whose labels exactly match the query label(s). This setting differs from ours, where queries impose multi-dimensional numeric constraints.
Range predicates enforce lower and/or upper bounds on numeric attributes. Most existing RFANNS methods for range predicates focus on single-attribute filters~\cite{engels2024approximate,zuo2024serf,xu2024irangegraph}, where the condition is specified on a single numeric dimension.
In contrast, we study RFANNS with multi-attribute numeric ranges, filling this gap.
For general comparison predicates, several systems support attribute-filtered ANN queries over heterogeneous attributes, including numeric, categorical, and other application-specific fields~\cite{mohoney2023high,patel2024acorn,wang2021milvus,wei2020analyticdb,xie2025beyond}.
Prior work has shown that different attribute types exhibit markedly different properties and that strong performance typically requires index designs tailored to a specific predicate class~\cite{xu2024irangegraph}.
Since our method is specifically designed for numeric range predicates, we therefore compare primarily against specialized RFANNS indexes rather than these general-purpose systems.
In addition, GPU-based methods for attribute-filtered ANN queries have recently been explored~\cite{xi2025vecflow}; these techniques are largely orthogonal to our contribution and could in principle be combined with our index.

\stitle{Dynamic range-filtering nearest neighbor search.}
Several wo-rks~\cite{jiang2025digra,zhang2025efficient,peng2025dynamic} have been proposed to support RFANNS in dynamic settings.
DIGRA~\cite{jiang2025digra} organizes objects in a B-tree over attributes and maintains HNSW graphs at tree nodes, both of which admit efficient maintenance under insertions and deletions.
DSG~\cite{peng2025dynamic} proposes a dynamic segment graph structure, while requiring only $O(\log n)$ new edges in expectation per insertion.

\section{Conclusion}
\label{sec:conclusion}

In this paper, we study RFANNS in high-dimensional Euclidean spaces with multi-attribute numeric range predicates, a setting that is increasingly important in modern vector databases but has been insufficiently explored by existing work.
We propose \kh, a tailored RFANNS index that combines an attribute-space partitioning tree with filtered HNSW graphs anchored at tree nodes.
The partitioning tree employs a skew-aware splitting strategy, which adapts to attribute distributions while still admitting provable bounds on the tree height.
On top of this structure, we design an efficient query algorithm that relies only on in-range neighbors yet still achieves high recall and strong query performance.

We conduct extensive experiments on four real-world datasets and compare \kh with \rg\ and a prefiltering baseline.
The results demonstrate that \kh consistently achieves better query throughput than the baselines.
Specifically, \kh achieves an average QPS speedup of $2.46\times$ over \rg\ on \laion, \dblp, and \ms, and $16.22\times$ on the more challenging \yt, with gains further increasing for smaller selectivities, larger $k$, and higher predicate cardinality.
When the recall is relaxed to $0.95$ (or $0.9$ on \yt), \kh also consistently outperforms \kw{Prefiltering} in QPS across all datasets, with an average speedup of $35.59\times$.
In terms of index construction, \kh is on average $3.27\times$ faster than \rg\ under a 16-thread configuration, while incurring at most roughly a $15\%$ space overhead compared with \rg.
Future work includes supporting fully dynamic workloads and exploring GPU implementations for large-scale deployments.


\clearpage
\balance
\bibliographystyle{ACM-Reference-Format}
\bibliography{ref}

\clearpage




\end{document}